\documentclass[a4paper,onecolumn,accepted=2026-04-19]{quantumarticle}
\pdfoutput=1
\usepackage{jorqu}

\author{Zackary Jorquera}
\affiliation{University of California, Santa Cruz, CA 95060, USA}
\email{zjorquera@ucsc.edu}
\author{Alexandra Kolla}
\affiliation{University of California, Santa Cruz, CA 95060, USA}
\author{Steven Kordonowy}
\affiliation{University of California, Santa Cruz, CA 95060, USA}
\author{Jaspreet Singh Sadhu}
\affiliation{University of California, Santa Cruz, CA 95060, USA}
\author{Stuart Wayland}
\affiliation{University of California, Santa Cruz, CA 95060, USA}

\title{Monogamy of Entanglement Bounds and Improved Approximation Algorithms for Qudit Hamiltonians}

\begin{document}

\pagenumbering{roman}

\maketitle

\begin{abstract}
    We prove new monogamy of entanglement bounds for two-local qudit Hamiltonians of rank-one projectors without one-local terms. In particular, we certify the maximum energy in terms of the maximum matching of the underlying interaction graph via low-degree sum-of-squares proofs. Algorithmically, we show that a simple matching-based algorithm approximates the maximum energy to at least $1/d$ for general graphs and to at least $1/d + \Theta(1/D)$ for graphs with bounded degree, $D$. This outperforms random assignment, which, in expectation, achieves energy of only $1/d^2$ of the maximum energy for general graphs. Notably, on $D$-regular graphs with degree, $D \leq 5$, and for any local dimension, $d$, we show that this simple matching-based algorithm has an approximation guarantee of $1/2$. Lastly, when \(d=2\), we present an algorithm achieving an approximation guarantee of \(0.595\), beating that of \cite{parekh_optimal_2022}, which gave an approximation ratio of \(1/2\).
\end{abstract}

\begingroup
\hypersetup{linktocpage}
\setcounter{tocdepth}{1}
\tableofcontents
\endgroup

\setcounter{page}{0}

\clearpage
\pagenumbering{arabic}

\section{Introduction}

A fruitful line of research in theoretical computer science has been to study classical constraint satisfaction problems (CSPs). These problems are generally \cc{NP}-hard, and so the question of the limits of efficient approximability of optimal solutions arises naturally. Under the widely-believed \algprobm{Unique Games Conjecture} (UGC), a semi-definite programming (SDP) based algorithm \cite{raghavendra_optimal_2008} that achieves the optimal approximation ratio is known.

The quantum analog of classical CSPs are (local) Hamiltonian problems. These problems are at the heart of condensed matter physics and quantum complexity theory, in which one is tasked with efficiently finding approximations of low-energy states of $n$-particle systems. Given a sequence of Hermitian operators $(H_i)_i$, that each act non-trivially on $k \leq n$ particles, the \algprobm{$k$-Local Hamiltonian} ($k$-LH) problem, when viewed as a maximization problem, is to find the most excited state of $H = \sum_{i} H_{i}$. Phrased as a decision problem, the \algprobm{$k$-Local Hamiltonian} problem is \cc{QMA}-complete \cite{kitaev_classical_2002}. Unlike the classical landscape of CSPs, there are still many unknowns about the approximability of general \algprobm{$k$-LH} including but not limited to the quantum PCP conjecture (see, for example, \cite{aharonov_quantum_2013,brandao_product-state_2016}). Attempting to study approximations to \algprobm{$k$-LH}s, \cite{brandao_product-state_2016} showed that interactions over high-degree graphs are well approximated by product states. This result leaves low-degree graphs, where entanglement is likely to play a larger role, as the primary avenue for study.

Attempting to understand the entanglement that may be present in excited states, we consider a family of 2-\algprobm{LH} with each term being a projector onto maximally entangled states (potentially weighted). A maximally entangled bipartite state can be specified uniquely (up to phase) by a \(d\)-dimensional special unitary. The \algprobm{Maximal Entanglement} problem is then defined as a tuple of an input graph, called the \emph{interaction graph}, and a sequence of qudit special unitaries, $\langle G = (V,E,w), (U_e)_{e \in E} \rangle$. Here, each $U_e$ then defines a rank-1 projector onto an arbitrary maximally entangled state, $\ket{\psi_e}\bra{\psi_e}$, from which we define the problem Hamiltonian $H = \sum_e w_e \ket{\psi_e}\bra{\psi_e}^e$. The task is to find the most excited state and/or estimate its energy. In this way, we view this problem as analogous to \algprobm{Unique Games} in which one is given a collection of edge permutations $(P_e : [d] \rightarrow [d])_{e \in E}$ and tasked to find an assignment $f : V \rightarrow [d]$ that satisfies as many of the permutations as possible\footnote{We say $P_{uv}$ is \textit{satisfied} by \(f : V \rightarrow [d]\) when $P_{uv}(f(v)) = f(u)$.}. Note that \algprobm{Unique Games} is not an instance of \algprobm{Maximal Entanglement} and \algprobm{Maximal Entanglement} is unrelated to the Unique Games generalizations considered in the non-local games literature \cite{kempe_unique_2010,mousavi_quantum_2024}. 
To further motivate the \algprobm{Maximal Entanglement} problem, we note that understanding the limits of entanglement arises naturally in the problem of finding excited states of frustrated systems\footnote{A LH is said to be \emph{frustrated} if the most excited state is not an eigenvector for all local terms.}. We propose the \algprobm{Maximal Entanglement} problem in an attempt to isolate this perspective and serve as a proxy for the limits of entanglement, as sums of rank-1 projectors onto maximally entangled bipartite states are naturally frustrated.

We certify that the energy of the most excited state for a \algprobm{Maximal Entanglement} instance is bounded above by the maximum matching of the underlying interaction graph. Here, the maximum matching is a subset of edges with no repeated vertices with maximum weighted sum. Such a matching can be found efficiently \cite{edmonds_maximum_1965}. Moreover, we show that a simple matching-based algorithm (i.e., doesn't use an SDP and instead the blossom algorithm from \cite{edmonds_maximum_1965}) achieves a non-trivial approximation ratio to the energy of the most excited state. We hope that these results serve as a proxy for lower and upper bounds on the amount of entangled present in a quantum state.

\subsection*{Related Work}

The \algprobm{Maximal Entanglement} problem, also referred to as the rank-one, strictly quadratic case of the \(2\)-LH problem \cite{parekh_beating_2021}, has been studied before in the qubit setting. When studying algorithms for this problem, one faces that so-called \emph{ansatz} problem, which asks how to write states succinctly and in a way that allows for quantities of interest to be efficiently calculated. Much of the existing algorithmic work in this field considers the mean-field model (i.e., that of product states) \cite{gharibian_approximation_2012,brandao_product-state_2016,parekh_beating_2021,parekh_optimal_2022}. Algorithmically, this has been a challenge for the general qudit setting \cite{carlson_approximation_2023}. Going beyond product state approximations, but still restricted to the qubit setting, \cite{anshu_improved_2021} applied low-degree circuits to approximate the global entanglement that might be present. When the local Hamiltonian problem is restricted to be the \algprobm{Quantum Max-Cut} Hamiltonian \cite{gharibian_approximation_2012}, products of 1 and 2 qubits states have been found to work well \cite{lee_improved_2024}, which is of particular interest to our work.

In particular, Lee and Parekh \cite{lee_improved_2024} show that a matching-based algorithm (i.e., one using Edmond's Algorithm \cite{edmonds_maximum_1965}) in combination with the Gharibian-Parekh algorithm \cite{gharibian_almost_2019} performs well on the \algprobm{Quantum Max-Cut} problem, achieving an approximation ratio of \(0.595\). For their analysis, they use results proven about the level-2 quantum Lasserre SDP specific to the qubit case or the \algprobm{Quantum Max-Cut} problem \cite{parekh_application_2021,parekh_optimal_2022}. We extend this work by arguing that such a matching-based algorithm also achieves non-trivial performance guarantees on the \algprobm{Maximal Entanglement} problem over qudit systems. When restricted to the qubit setting, we match this approximation ratio, achieving a guarantee of \(0.595\) to the energy of the most excited state.

In the qubit setting there are some special cases of the \algprobm{Maximal Entanglement} problem, such as the \algprobm{EPR} problem \cite{king_improved_2023} and \algprobm{Quantum Max-Cut} \cite{gharibian_almost_2019}, that have been used as test beads for algorithmic design.
Since the submission of this manuscript, there has been a sequence of results finding improved approximation algorithms for these two problems. We briefly summarize these results in the following table.

\begin{table}[ht!]
    \centering

    \begin{minipage}{0.45\textwidth}
    \centering
        {\renewcommand{\arraystretch}{1.25}
        \begin{tabular}{|c|c|}
        \hline
        Paper & Approximation Guarantee \\ \hline \hline
        \cite{king_improved_2023}     & $\frac{1}{\sqrt{2}} \approx 0.707$                   \\ \hline
        \textbf{This work}     & $\bm{\frac{18}{25} = 0.72}$                   \\ \hline
        \cite{ju2025improvedapproximationalgorithmsepr,apte2025improvedalgorithmsquantummaxcut}     & $\frac{1 + \sqrt{5}}{2} \approx 0.809$                   \\ \hline
        \cite{apte2025_EPR}     & $0.8395$                   \\ \hline
        \end{tabular}}
        \label{tab:epr}

    \vspace{2pt}
    (a) \algprobm{EPR}
    \end{minipage}
    \hfill
    \begin{minipage}{0.45\textwidth}
    \centering
    \begin{tabular}{|c|c|}
        \hline
        Paper & Approximation Guarantee \\ \hline \hline
        \cite{gharibian_almost_2019}     & $0.498$                   \\ \hline
        \cite{parekh_application_2021}     & $0.5$                   \\ \hline
        \cite{huber2024secondorderconerelaxations} & $0.526$                   \\ \hline
        \cite{anshu_beyond_2020}     & $0.531$                   \\ \hline
        \cite{parekh_application_2021}     & $0.533$                   \\ \hline
        \cite{lee2022optimizingquantumcircuitparameters}    & $0.562$                   \\ \hline
        \cite{lee_improved_2024}     & $0.595$                   \\ \hline
        \textbf{This work}     & $\bm{0.599}$                   \\ \hline
        \cite{gribling2025improvedapproximationratiosquantum}     & $0.603$                   \\ \hline
        \cite{apte2025improvedalgorithmsquantummaxcut}    & $0.611$                   \\ \hline
        \end{tabular}
    
    \vspace{2pt}
    (b) \algprobm{Quantum Max-Cut}
    \end{minipage}
    
    \caption{History of known results on two special types of qubit \algprobm{Maximal Entanglement} instances: (a) \algprobm{EPR} and (b) \algprobm{Quantum Max-Cut}. Note, \cite{ju2025improvedapproximationalgorithmsepr,apte2025improvedalgorithmsquantummaxcut,gribling2025improvedapproximationratiosquantum,apte2025_EPR} were all announced after the present paper was originally submitted. Up-to-date results for these problems are tracked at \cite{qmc_reference_site}.
    }
    \label{tab:overall}
\end{table}

\subsection*{Our Results}

We first give our results over general interaction graphs.
\begin{manualtheorem}{A}[\cref{thm_upper_bound_on_QUG_energy,thm_main_genral_graph}]\label{thm_A}
    For any instance of the \algprobm{Maximal Entanglement} problem over an interaction graph, \(G\), with \(H = \E_{e \sim E(G)} \ket{\psi_e}\bra{\psi_e}^e\) the normalized problem Hamiltonian, there exists an efficient algorithm that outputs a density matrix $\rho$ such that $\Tr(\rho H) \geq \frac{1}{d}\Tr(\rho_* H)$, where $\rho_*$ is a most excited state. Furthermore, there exist low-degree sum-of-squares certificates certifying that, for any quantum state \(\rho\), $\Tr(\rho H) \leq \frac{1}{d} + \frac{5(d-1)}{4d}\OPT_{\mathsc{Match}}(G)$, where \(\OPT_{\mathsc{Match}}(G)\) is the maximum matching of \(G\).
\end{manualtheorem}

This is significant because the approximation ratio of random assignment (i.e., the maximally mixed state, \(\rho^{\mathsc{mixed}}\)) is only such that \(\Tr(\rho^{\mathsc{mixed}} H) \geq \frac{1}{d^2}\Tr(\rho_* H)\). In this sense, the algorithm of \cref{thm_A} beats random assignment. We can be more precise about this and show that this algorithm beats random assignment even when we restrict to bounded degree graphs or regular graphs. In this setting random assignment gives \(\Tr(\rho^{\mathsc{mixed}} H) \geq \left(\frac{1}{d} - \Theta(\frac{1}{D})\right)\Tr(\rho_* H)\) (see, for example, \cref{remark_sdasdhucakjSAEkub}).\zacktodo{It would be nice to reference something like \cref{prop_prod_and_max_mixed_approx}.}

\begin{manualtheorem}{B}[\cref{thm_main_bounded_min_max_deg_graph,cor_main_bounded_deg_5}]\label{thm_B}
    In the situation of the above theorem and when \(G\) is unweighted and bounded in degree by $D$, there exists an efficient algorithm that outputs a density matrix $\rho$ such that $\Tr(\rho H) \geq \left(\frac{1}{d} + \Omega(\frac{1}{D})\right)\Tr(\rho_* H)$. Furthermore, over regular graphs of degree \(D \leq 5\) we have that \(\Tr(\rho H) > \frac{1}{2}\Tr(\rho_* H)\) for all \(d \geq 2\).
\end{manualtheorem}

The algorithm in \cref{thm_A,thm_B} is a simple matching based algorithm (adapted from \cite{lee_improved_2024}), which does not use an SDP. However, the analysis requires monogamy of entanglement results that we prove with degree-6 sum-of-squares certificates or, equivalently, the level-3 quantum Lasserre SDP. Indeed, these certificates, given in \cref{thm_C}, can be seen as our key contribution, and we hope they will be of independent interest. We note that, in general, these bounds are tight (see, for example, \cref{thm_sos_epr_star_bound_is_opt,remark_tri_bound_is_opt_EPR}).

\begin{manualtheorem}{C}[\cref{prop_QUG_SOS_star_bound,prop_QUG_SOS_triangle_bound}]\label{thm_C}
    In the situation of \cref{thm_A}, let $\rho$ be an arbitrary density matrix. For any vertex $a \in V$ and any subset of its neighbors $S \subseteq N(a)$, 
    \[\Tr\left(\rho \sum_{b \in S}\ket{\psi_{ab}}\bra{\psi_{ab}}^{ab}\right) \leq \frac{\abs{S}}{d} + \frac{d-1}{d}\text{.}\]
    Moreover, for any triangle $\{a,b,c\} \subseteq V$,
    \[\Tr\left(\rho \left(\ket{\psi_{ab}}\bra{\psi_{ab}}^{ab} + \ket{\psi_{ac}}\bra{\psi_{ac}}^{ac} + \ket{\psi_{bc}}\bra{\psi_{bc}}^{bc}\right)\right) \leq \frac{3}{d} + \frac{d-1}{d}\text{.}\]
\end{manualtheorem}

As a special case of the \algprobm{Maximal Entanglement} problem, when working over qubits (\(d=2\)) we show that the matching-based algorithm from \cref{thm_A} combined with the Parekh-Thompson product state algorithm \cite{parekh_beating_2021} achieves an approximation guarantee of \(0.595\). This beats the previously best known algorithm for this 2-LH, which had an approximation guarantee of \(\frac{1}{2}\) \cite{parekh_optimal_2022}. Additionally, in the case of \algprobm{Quantum Max-Cut}, we give slightly improved analysis to show that the algorithm in \cite{lee_improved_2024} achieves an approximation guarantee of \(0.599\), where it was previously only known to have an approximation guarantee of \(0.595\). Lastly, for the qubit \algprobm{EPR} problem we give an algorithm that achieves an approximation guarantee of \(0.72\) beating that of \cite{king_improved_2023}, which gave an algorithm with an approximation guarantee of \(\frac{1}{\sqrt{2}}\). These result are all delegated to \cref{section_the_qubit_case}.

\subsection*{Significance}

We believe our work can help better understand entanglement in arbitrary quantum states. Since the sum-of-squares certificates apply equally to true quantum states, an equivalent formulation of our results is that for an arbitrary state, we can characterize the ``amount of entanglement'' over an edge in the interaction graph by considering the supremum of the energy overall projectors onto maximally entangled states. In particular, \cref{thm_C} can be seen as giving new monogamy of entanglement style bounds. Previously, these bounds were only known for the qubit case \cite{anshu_beyond_2020,parekh_application_2021} or more restrictively, for anti-symmetric entanglement, i.e., when considering the triangle graph \cite{parekh_optimal_2022}. Globally, \cref{thm_A,thm_B} bound the expected entanglement over edges of a \(D\)-regular graph by \(1/d+\calO(1/D)\), for constant local dimension, \(d\). For cases when there exists an optimal product state approximation, i.e., the EPR problem, this beats the \(1/d+\calO(1/D^{1/3})\) upper bound achieved by \cite{brandao_product-state_2016}.

Algorithmically, we demonstrate that the matching-based algorithm achieves at least a constant factor of this upper bound, lower bounding the maximum energy by \(1/d^2 + \Omega(1/D)\). A natural follow-up question is if the certificates can be improved or if the algorithm providing product state witnesses can be improved. If the answer to both these questions is no, then it is possible that some parameterized family of instances of the \algprobm{Maximal Entanglement} problem are, for instance, candidate NLTS instances. A resolution to this (either way) would further help understand the relationship between different types of entanglement and the circuit complexity of generating said entangled states.

\subsection{Preliminaries And Notation}\label{section_prelims}

We use the notation \([n] \coloneq \{1, \dotsc, n\}\). We denote the \textit{standard basis} of \(\C^d\) as \(\{\ket{i}\ |\ i \in [d]\}\). For, $A \in \mathcal{L}((\C^d)^{\otimes n})$, a \emph{bounded linear operator} from \((\C^d)^{\otimes n}\) to itself, we use \(A^\sfT\) to denote the \emph{transpose} and \(A^\dagger \coloneq \overline{A^\sfT}\) to denote the \emph{adjoint/conjugate transpose}. \(A\) is \emph{Hermitian} if \(A^\dagger = A\) and it is a \emph{projector} if \(A^2 = A\). We use the notation \(\calD((\C^d)^{\otimes n}) \coloneq \{\rho \in \calL((\C^d)^{\otimes n})\ |\ \rho^\dagger = \rho,\ \rho \succcurlyeq 0,\ \Tr(\rho)=1\}\) to denote the subset of \emph{density matrices} on \(n\) qudits. A density matrix $\rho$ is called \emph{pure} if it is a projector, namely, \(\rho^2 = \rho = \ket{\psi}\bra{\psi}\) for some \(\ket{\psi} \in (\C^d)^{\otimes n}\). We use \(\{A,B\} = AB + BA\) to denote the \emph{anti-commutator}.

We use superscripts in two ways. First, let \(\rho \in \calD((\C^d)^{\otimes n})\) be a quantum state over \(n\) qudits, then we use the superscript notation to denote its reduced density matrices. That is, let \(S \subseteq [n]\), then \(\rho^S \coloneq \Tr_{[n] \setminus S}(\rho)\) and when \(S = \{a\}\), we will use \(\rho^a\). Second, when \(A \in \calL((\C^d)^{\otimes k})\) is a linear operator over \(k<n\) qudits, we will use the superscript notation with a sequence of non-repeating indices, \((a_1, \dotsc, a_k)\), to extend it to an operator over \(n\) qudits, \(A^{a_1 \dotsc a_k} \in \calL((\C^d)^{\otimes n})\), where the indices specify which qudits to apply the operator with all other qudits being acted on by the identity. Sometimes, to be more explicit, we will use \(A^{a_1 \dotsc a_k} \otimes I^{[n] \setminus \{a_1, \dotsc, a_k\}} \in \calL((\C^d)^{\otimes n})\). We do this for kets as well to specify the order of tensors, e.g., \(\ket{\psi}^{23} \otimes \ket{\varphi}^1 \coloneq \ket{\varphi} \otimes \ket{\psi} \in (\C^d)^{\otimes 3}\). In this setting we will never leave out indices. 

We make substantial use of the \emph{generalized EPR state},
\(
    \ket{\EPR_d} \coloneq \frac{1}{\sqrt{d}} \sum_{a = 1}^d \ket{a} \otimes \ket{a}
\in (\C^d)^{\otimes 2}\). 
When the local dimension is implied, we simplify to \(\ket{\EPR}\). We also let \(\EPR \coloneq \ket{\EPR}\bra{\EPR} \in \calL((\C^d)^{\otimes 2})\) denote the projector onto the subspace span by \(\ket{\EPR}\). 

For an algorithm \(\calA\), an \emph{approximation guarantee} is a constant \(\alpha\) such that for all problem instances \(\calI\), one has that \(\Tr(\rho H) \geq \alpha \Tr(\rho_* H)\) for \(\rho = \calA(\calI)\), the state output by the algorithm, \(H = H(\calI)\), the Hamiltonian defined by the instance, and \(\rho_*\) the optimal solution for \(\calI\)/most excited state of \(H\). We typically refer to the \emph{approximation ratio} (or \emph{algorithmic gap}) as the infimum over all instances, \(\alpha = \inf_\calI\left(\frac{\Tr(\rho H)}{\Tr(\rho_* H)}\right)\).

For two functions \(f,g : \calD \to \RR\) and \(a \in \calD\) we say \(f(x) \in \Theta_{x \to a}(g(x))\) if the limit \(C = \lim_{x \to a}\left(\frac{f(x)}{g(x)}\right)\) exists, is non-zero, and is bounded. When \(a = \infty\) we will drop the subscript and just write \(f(x) \in \Theta(g(x))\). We often use the notation \(f(x) \geq k(x) \pm \Theta_{x \to a}(g(x))\) for some explicit function \(k: \calD \to \R\) to mean that there exists a function \(h : \calD \to \RR\) such that \(\forall x \in \calD : f(x) \geq h(x)\) and \(\pm h(x) \mp k(x) \in \Theta_{x \to a}(g(x))\). We similarly use Big-\(\calO\) and Big-\(\Omega\).

\section{The \texorpdfstring{\textsc{Maximal Entanglement}}{Maximal Entanglement} Problem}

In this paper, we consider a subclass of the \algprobm{$2$-Local Hamiltonian} problem in which the edge interactions are projectors onto maximally entangled states over \(d\)-dimensional qudits. Because our problem is \(2\)-local, there exists a natural underlying interaction graph.

\begin{definition}[\algprobm{Maximal Entanglement} Problem]\label{def_QUG_prob}
    Over an \(n\)-qudit system of local dimension \(d\), we are given a positively weighted graph \(G = (V,E,w)\) with \(|V| = n\) and a sequence of unitary matrices, \((U_e \in \SU(d))_{e \in E}\), indexed by the edges and specified by \(\poly(n)\) bits. We then define a \(2\)-local Hamiltonian problem with local Hamiltonians defined by \(h_e \coloneq (I \otimes U_e) \ket{\EPR}\bra{\EPR} (I \otimes U_e^\dagger)\)\footnote{Note, we implicitly define a total ordering on \(V\) and apply the unitary to the qudit associated with the maximal vertex in the edge according to this ordering.} and the full (normalized) Hamiltonian \(H \coloneq \E_{(a,b) \sim E} \left[h_{ab}^{a b} \otimes I^{[V] \setminus \{a,b\}}\right] = \frac{1}{W} \sum_{(a,b) \in E} w_{ab} \left(h_{ab}^{a b} \otimes I^{[V] \setminus \{a,b\}}\right)\) (where \(W \coloneq \sum_{(a,b) \in E} w_{ab}\)). As an optimization problem, we have the following objective:
    \begin{equation}\label{eq_QUG_obj}
        \lambda_{\max}(H)=\max_{\ket{\psi} \in (\C^d)^{\otimes n}} \bra{\psi} H \ket{\psi} = \max_{\rho \in \calD((\C^d)^{\otimes n})} \Tr(H \rho)
    \end{equation}
\end{definition}

There are two special instances of the \algprobm{Maximal Entanglement} (ME) problem problem worth highlighting. First, when all the unitaries are the identity matrix (\(U_{e} = I\) for all \(e \in E\)) this is known as the \algprobm{EPR} problem. This problem is known to be stoquastic, i.e., as a minimization problem, its Hamiltonian (negative of the one we consider), written in the standard basis, has non-positive off-diagonal elements. It is well known that these Hamiltonians and indeed the \algprobm{EPR} problem belong to the class \cc{StoqMA}. It is believed that \(\cc{StoqMA} \subsetneq \cc{QMA}\) \cite{bravyi_merlin-arthur_2006,cubitt_complexity_2014}. Secondly, when working with qubits (\(d=2\)) and if every edge unitary is taken to be \(iY\) (where \(Y\) is the Pauli-\(Y\) matrix), then this problem becomes \algprobm{Quantum Max-Cut} \cite{gharibian_almost_2019}, which is \cc{QMA}-hard as an optimization problem. In particular, the \algprobm{Maximal Entanglement} problem is also \cc{QMA}-hard as an optimization problem.

To give more context for and to justify \cref{def_QUG_prob}, we look at some well-known properties of maximally entangled states and discuss their implications on the \algprobm{Maximal Entanglement} problem.

\begin{definition}[Maximally Entangled State]\label{def_max_ent_state}
    A bipartite state, \(\rho \in \calD((\C^d)^{\otimes 2})\), is said to be pure if \(\rho^2 = \rho\) and \emph{maximally entangled} if its reduced density matrices are maximally mixed, i.e., \(\Tr_1(\rho) = \Tr_2(\rho) = \frac{1}{d}I\).
\end{definition}

Let \(\ket{\psi} \in (\C^d)^{\otimes 2}\) be an arbitrary bipartite state with global phase. By the Schmidt decomposition, there exists orthonormal bases \(\{\ket{e_1}, \dotsc, \ket{e_d}\}\) and \(\{\ket{f_1}, \dotsc, \ket{f_d}\}\) for \(\C^d\) along with non-negative constants \((\lambda_i)_i\) such that
\(\ket{\psi} = \sum_{i = 1}^d \sqrt{\lambda_i} \ket{e_i} \otimes \ket{f_i}\).
The reduced density matrices are then \(\psi^1 = \sum_{i = 1}^m \lambda_i \ket{e_i}\bra{e_i}\) and \(\psi^2 = \sum_{i = 1}^m \lambda_i \ket{f_i}\bra{f_i}\). The $\psi^i$ are both the maximally mixed state exactly when \(\lambda_i = \frac{1}{d}\) for all \(i \in [d]\). This is all to say that a state is maximally entangled if and only if it can be written as \(\ket{\psi} = \frac{1}{\sqrt{d}}\sum_{i = 1}^d \ket{e_i} \otimes \ket{f_i}\) for some orthonormal bases \(\{\ket{e_1}, \dotsc, \ket{e_d}\}\) and \(\{\ket{f_1}, \dotsc, \ket{f_d}\}\). Additionally, if \(\lambda_i > 0\) for all \(i \in [d]\), then $\ket{\psi}$ has \emph{full Schmidt rank}.

\begin{proposition}[Facts About Full-Rank States]\label{prop_full_rank_states}
    Let \(\ket{\psi} \in (\C^d)^{\otimes 2}\) be a bipartite state with full Schmidt rank. We use \(GL(d)\) to denote the \emph{general linear group} over \(\C^d\) and \(\calU(d) \subseteq GL(d)\) the \emph{unitary group}.
    \begin{enumerate}
        \item\label{prop_item_full_rank_states_2} There exists a unique \(A \in GL(d)\) such that \(\ket{\psi} = (I \otimes A) \ket{\EPR}\). In particular, if \(\ket{\psi}\) is maximally entangled then \(A \in \calU(d)\).
        
        \item\label{prop_item_full_rank_states_1} For any \(A \in GL(d)\) we have that there exists a unique \(B \in GL(d)\) such that \((I \otimes A) \ket{\psi} = (B \otimes I) \ket{\psi}\) and vise versa. In particular, if \(\ket{\psi}\) is maximally entangled and \(A \in \calU(d)\), then \(B \in \calU(d)\). Additionally, if \(\ket{\psi} = \ket{\EPR}\) then \(B=A^{\sfT}\).
    \end{enumerate}
\end{proposition}

While well known\footnote{The special case of \cref{prop_item_full_rank_states_1}, when \(\ket{\psi} = \ket{\EPR}\), is commonly known as the transpose-trick.}, we include a proof of this proposition in \cref{appendix_proof_of_prop_full_rank_states} for completeness. Next, we have the following corollary, which adapts \cref{prop_full_rank_states} for pure density matrices (i.e., states without global phase).

\begin{corollary}\label{cor_max_ent_states} 
Let \(\SU(d) \subseteq \calU(d)\) denote the \emph{special unitary group}. Let \(\rho \in \calD((\C^{d})^{\otimes 2})\) be a maximally entangled pure state.

\begin{enumerate}
    \item There exists a unique (up to phase) \(U \in \SU(d)\) such that \(\rho = (I \otimes U) \ket{\EPR}\bra{\EPR} (I \otimes U^\dagger)\). Moreover, for every \(U \in \SU(d)\),  \((I \otimes U) \ket{\EPR} \bra{\EPR} (I \otimes U^\dagger)\) is a maximally entangled state. 
    \item For every \(A \in \SU(d)\), there exists a unique (up to phase) \(B \in \SU(d)\) such that \((A \otimes I) \rho (A^\dagger \otimes I) = (I \otimes B) \rho (I \otimes B^\dagger)\).
\end{enumerate}
\end{corollary}

\begin{proof}
    By \cref{prop_full_rank_states}, we know that for every maximally entangled pure state \(\rho \in \calD((\C^{d})^{\otimes 2})\) there exists a \(\ket{\psi} \in (\C^{d})^{\otimes 2}\) such that \(\rho = \ket{\psi}\bra{\psi}\) and a unique \(U \in U(d)\) such that \(\ket{\psi} = (I \otimes U) \ket{\EPR}\), i.e., \(\rho = (I \otimes U) \ket{\EPR}\bra{\EPR} (I \otimes U^\dagger)\). Let \(z = \det(U)\), we then have that \(\sqrt[d]{\overline{z}}U \in \SU(d)\) and \((I \otimes \sqrt[d]{\overline{z}}U) \ket{\EPR}\bra{\EPR} (I \otimes \overline{\sqrt[d]{\overline{z}}} U^\dagger) = (I \otimes U) \ket{\EPR}\bra{\EPR} (I \otimes U^\dagger) = \rho\). The other direction follows from Schmidt decomposition directly.

    The second statement can be shown by modifying the proof of \cref{prop_full_rank_states}, \cref{prop_item_full_rank_states_1} with careful consideration of the determinant of \(B\) when using the fact that \(U \in \SU(d)\) in the decomposition \(\rho = \ket{\psi}\bra{\psi}\) and \(\ket{\psi} = (I \otimes U) \ket{\EPR}\). 
    
    The uniqueness up to phase in both parts follows from the uniqueness is \cref{prop_full_rank_states}, which is relaxed by the fact that conjugation is invariant under change of phase.
\end{proof}

With this, it follows that the local Hamiltonians in \cref{def_QUG_prob} are nothing but projectors onto arbitrary maximally entangled bipartite states as wanted. Furthermore, without loss of generality, the unitary can be applied to either qudit in the edge (up to transpose). Moreover, \cref{cor_max_ent_states} allows us to reason about solutions to the \algprobm{Maximal Entanglement} problem over certain interaction graphs with solutions to the EPR problem (or the \algprobm{Quantum Max-Cut} problem in the case of \(d=2\)).

\subsection{A Priori Analysis}

First, we show that any instance of the \algprobm{Maximal Entanglement (ME)} problem over a tree graph is equivalent to the \algprobm{EPR} problem over the same tree graph, in the following sense. 

\begin{lemma}[\algprobm{ME} Is Equivalent To \algprobm{EPR} On Tree Graphs]\label{lemma_equiv_of_QUG_EPR}
    Let \(\langle T = (V,E,w), (U_e)_e \rangle\) be an instance of the \algprobm{Maximal Entanglement} problem where \(T\) is a tree graph (i.e., no cycles).
    Then, for any \(\rho \in \calD((\C^d)^{\otimes n})\) we have that there exists a \(\rho' \in \calD((\C^d)^{\otimes n})\) such that, for all \((a,b) \in E\),
    \begin{equation}
        \Tr\left(h_{ab}^{ab}\rho\right) = \Tr\left(\EPR^{ab}_{ab} \rho'\right)
    \end{equation}
    and vice versa (where \(h_{ab}\) is defined as in \cref{def_QUG_prob}).
\end{lemma}

\begin{proof}
    Fix some root vertex \(r\). We will assume, without loss of generality, that the unitaries for the ME instance are always applied to the vertices with greater depth.
    We inductively construct a local unitary over the depth of the tree. Let \(U_r = I\). For each vertex \(a \in V\), such that \(a\) is a child of \(r\) (i.e., has depth 1) let \(U_{a} = U_{(r,a)}\), where \(U_{(r,a)}\) is the unitary in the ME instance assigned to the edge \((r,a)\). For each vertex, \(a \in V\), with depth \(2\), let \(b \in V\) be it parent vertex, then let \(U_{a} = U_{b,a} \overline{U_{b}}\), where \(\overline{U_{b}}\) is the unitary guaranteed by \cref{cor_max_ent_states} to have the property that \((U_b \otimes \overline{U_{b}}) \EPR (U_b \otimes \overline{U_{b}})^\dagger = \EPR\). We continue this process inductively, for increasing depths, so we have the following state with the desired property.
    \begin{equation}\label{eq_conj_by_loc_units_to_get_eqiv_state}
        \rho' = \left(\bigotimes_{a \in V} \left(U_{a}^\dagger\right)^{a}\right) \rho \left(\bigotimes_{a \in V}U_{a}^{a}\right)
    \end{equation}
    The claim then follows by the cyclic property of the trace.
\end{proof} 

\begin{figure}[t]
    \centering
    \scalebox{0.7}{\begin{tikzpicture}[scale=2]
    \node[circle, draw, fill=orange30] (1) at (0,0) {$1$};
    \node[circle, draw, fill=orange30] (2) at ($(1) + (1*360/5 + 54: 1)$) {$2$};
    \node[circle, draw, fill=orange30] (3) at ($(1) + (2*360/5 + 54: 1)$) {$3$};
    \node[circle, draw, fill=orange30] (4) at ($(1) + (3*360/5 + 54: 1)$) {$4$};
    \node[circle, draw, fill=orange30] (5) at ($(1) + (4*360/5 + 54: 1)$) {$5$};
    \node[circle, draw, fill=orange30] (6) at ($(1) + (5*360/5 + 54: 1)$) {$6$};

    \node[circle, draw, fill=purple20] (1b) at (4,0) {$1$};
    \node[circle, draw, fill=purple20] (2b) at ($(1b) + (1*360/5 + 54: 1)$) {$2$};
    \node[circle, draw, fill=purple20] (3b) at ($(1b) + (2*360/5 + 54: 1)$) {$3$};
    \node[circle, draw, fill=purple20] (4b) at ($(1b) + (3*360/5 + 54: 1)$) {$4$};
    \node[circle, draw, fill=purple20] (5b) at ($(1b) + (4*360/5 + 54: 1)$) {$5$};
    \node[circle, draw, fill=purple20] (6b) at ($(1b) + (5*360/5 + 54: 1)$) {$6$};

    \draw (2) -- (1)node[midway,above,sloped,text=green!50!black] {$h_{12}$} -- (6)node[midway,above,sloped,text=green!50!black] {$h_{16}$}
    (3) -- (1)node[midway,above,sloped,text=green!50!black] {$h_{13}$} 
    (4) -- (1)node[midway,above,sloped,text=green!50!black] {$h_{14}$} -- (5)node[midway,above,sloped,text=green!50!black] {$h_{15}$};

    \draw (2b) -- (1b)node[midway,above,sloped,text=blue!50!black] {\small $\operatorname{EPR}$} -- (6b)node[midway,above,sloped,text=blue!50!black] {\small $\operatorname{EPR}$}
    (3b) -- (1b)node[midway,above,sloped,text=blue!50!black] {\small $\operatorname{EPR}$} 
    (4b) -- (1b)node[midway,above,sloped,text=blue!50!black] {\small $\operatorname{EPR}$} -- (5b)node[midway,above,sloped,text=blue!50!black] {\small $\operatorname{EPR}$};

    \draw ($(1)+(1,0.5)$) edge[bend left,<->,line width=1] ($(1b)+(-1,0.5)$);
    \node at ($(1)+(2,1)$) {\large $\Tr(\text{\textcolor{orange50}{$\rho$}} \text{\textcolor{green!50!black}{$H$}}) = \Tr(\text{\textcolor{purple50}{$\rho'$}} \text{\textcolor{blue!50!black}{$H^{\operatorname{EPR}}$}})$};
\end{tikzpicture}}
    \caption{A visual representation of \cref{lemma_equiv_of_QUG_EPR} on the star graph.
    }
    \label{fig_ME_equals_EPR_star}
\end{figure}

While this property will turn out to be vastly important, it is not sufficient for our analysis as we will need to consider interaction graphs where we can not, in general, show this type of equivalence (for example, the triangle graph). To circumvent this, we give another very important property that will turn out to be of great importance for our analysis. First, we give the following lemma in terms of the EPR problem. Then, we will argue that the same property is true for arbitrary projectors onto maximally entangled states.

\begin{lemma}\label{lemma_anticom_of_epr_give_sum_and_proj}
There exists a projector \(P \in \calL((\C^d)^{\otimes 3})\) such that
    \begin{equation}\label{eq_anticom_of_epr_give_sum_and_proj}
        \{\EPR \otimes I, I \otimes \EPR\} = \frac{1}{d}(\EPR \otimes I + I \otimes \EPR) - \frac{2(d-1)}{d^2} P\text{.}
    \end{equation}
\end{lemma}

The proof of this lemma is done by direct calculation and given in \cref{appendix_proof_of_lemma_anticom_of_epr_give_sum_and_proj}. We note that in the \(d=2\) case, \(P=\Pi_{\vcenter{\hbox{\scalebox{0.4}{$\ydiagram{1,1}$}}}}^{13}\otimes I^2\) is the projector onto the anti-symmetric subspace over the first and third subsystem and identity over the second (i.e., the edge interaction for \algprobm{Quantum Max-Cut} \cite{gharibian_almost_2019}). For \(d > 2\), we note that \(P\) does not, in general, act trivially on the second subsystem. We then have the following corollary.

\begin{corollary}\label{cor_anticom_of_UGC_give_sum_and_proj}
    Let \(h_1\) and \(h_2\) be projectors onto maximally entangled 2-qudit states. There exists a projector \(P \in \calL((\C^d)^{\otimes 3})\)  such that
    \begin{equation}\label{eq_anticom_of_QUG_give_sum_and_proj}
        \{h_1 \otimes I, I \otimes h_{2}\} = \frac{1}{d}(h_{1} \otimes I + I \otimes h_{2}) - \frac{2(d-1)}{d^2} P\text{.}
    \end{equation}
\end{corollary}

\begin{proof}
    By \cref{cor_max_ent_states}, let \(U_{1}, U_{2} \in \SU(d)\) be such that \(h_{1} = (U_{1} \otimes I) \ket{\EPR}\bra{\EPR} (U^\dagger_{1} \otimes I)\) and \(h_{2} = (I\otimes U_{2}) \ket{\EPR}\bra{\EPR} (I \otimes U^\dagger_{2})\). We then conjugate both sides of \cref{eq_anticom_of_epr_give_sum_and_proj} by \((U_{1} \otimes I \otimes U_{2})\) to get the following for \(P\) as in \cref{lemma_anticom_of_epr_give_sum_and_proj}.
    \begin{align}
        \{h_{1} \otimes I, I \otimes h_{2}\} = \frac{1}{d}(h_{1} \otimes I + I \otimes h_{2}) - \frac{2(d-1)}{d^2} (U_{1} \otimes I \otimes U_{2})P(U^\dagger_{1} \otimes I \otimes U^\dagger_{2})
    \end{align}
    We note \((U_{1} \otimes I \otimes U_{2})P(U^\dagger_{1} \otimes I \otimes U^\dagger_{2})\) is a projector because \(P\) is one.
\end{proof}

\begin{remark}
    We note that \cref{cor_anticom_of_UGC_give_sum_and_proj} can be further generalized to the case when \(h_1 = \sum_a \alpha_a \ket{\psi_a} \bra{\psi_a}\) and \(h_2 = \sum_a \beta_a \ket{\varphi_a} \bra{\varphi_a}\) are convex combinations of rank-one projectors and \(P\) is positive semi-definite. 
\end{remark}

To our knowledge, we are the first to discover this property at this level of generality (i.e., for \(d \geq 2\)). If we again consider the \(d=2\) case and let \(U_{1} = w U_{2}^{\sfT}\) for some phase, \(w \in \C\) with \(|w| = 1\), (which is the case for when \(h_{1} = h_{2} = \Pi_{\vcenter{\hbox{\scalebox{0.4}{$\ydiagram{1,1}$}}}}\) are the edge interaction for \algprobm{Quantum Max-Cut} as \(U_{1}\) and \(U_{2}\) can be taken to be \(iY\)) we have that \(P=\Pi_{\vcenter{\hbox{\scalebox{0.4}{$\ydiagram{1,1}$}}}}^{13}\otimes I^2\) is again the projector onto the anti-symmetric subspace over the first and third subsystem and identity over the second. This fact, for \algprobm{Quantum Max-Cut}, has been previously discovered \cite{parekh_optimal_2022,wright_sosproofsemail_2023,takahashi_su2-symmetric_2023}.

\section{Certificates for Monogamy of Entanglement Bounds}

We first set the groundwork for using the sum-of-squares (SOS) proof technique.
To talk about the (non-commutative) sum-of-squares hierarchy, we must first define a notion of degree. We do this by picking a basis for \(\calL(\C^d)\), with one element being the identity, denoted \(\calB_d = \{I\} \cup \{T_a\ |\ a \in [d^2-1]\}\). We will often denote the identity element using \(T_0 \coloneq I\). This can be extended to a basis for \(\calL((\C^d)^{\otimes n}) \cong (\calL(\C^d))^{\otimes n}\),
\begin{equation}
    \calB_d^n = \left\{\bigotimes_{a = 1}^n T_{b_a-1}\ \middle|\ \forall a \in [n] : b_a \in [d^2]\right\}
\end{equation}

\begin{definition}
    For a basis element \(T \in \calB_d^n\), we define its \emph{degree}, denoted \(\omega(T)\), to be the number of non-identity terms or the number of systems on which \(T\) acts non-trivially. We extend this notion to an arbitrary operator to be the largest degree among all of its non-zero components.
\end{definition}

It is clear that we can bound the degree of an arbitrary operator by the number of qudits such that it acts non-trivially. We note that, while defined using \(\calB_d\), this definition is independent of the choice of \(\calB_d\), up to requiring the identity element.

With this, we can define the SOS hierarchy and the notion of a pseudo-density matrix, which mirrors the notation of a pseudo-distribution in the commutative setting \cite{barak_proofs_2016}.

\begin{definition}[Degree-\(2t\) Pseudo-Density Matrix]\label{def_pdm}
    An operator \(\Tilde{\rho} \in \mathcal{L}((\C^d)^{\otimes n})\) is called a \emph{degree-\(2t\) pseudo-density matrix} over \(n\), \(d\)-dimensional qudits if it is normalized, i.e., \(\Tr(\Tilde{\rho}) = 1\), and lies in the closed convex cone defined by:
    \begin{enumerate}
        \item\label{def_item_pdm_self_adj} Self-Adjoint: \(\Tilde{\rho}^\dagger = \Tilde{\rho}\), and
        \item\label{def_item_pdm_positivity} Positivity:  \(\Tr(\Tilde{\rho} A^\dagger A) \geq 0\) for all \(A \in \mathcal{L}((\C^d)^{\otimes n})\) such that \(\omega(A) \leq t\).
    \end{enumerate} 
    Additionally, we use \(\Tilde{\mathcal{D}}^{(2t)}((\C^d)^{\otimes n})\) to denote all such degree-\(2t\) pseudo-density matrices.
\end{definition}

\begin{proposition}[The SoS/Pseudo-Density Matrix Hierarchy]\label{prop_degree_2n_pdm_is_true_dm}
    We have the following hierarchy.
    \begin{equation*}
        \Tilde{\mathcal{D}}^{(2)}((\C^d)^{\otimes n}) \supset \Tilde{\mathcal{D}}^{(4)}((\C^d)^{\otimes n}) \supset \cdots \supset\Tilde{\mathcal{D}}^{(2(n-1))}((\C^d)^{\otimes n}) \supset \Tilde{\mathcal{D}}^{(2n)}((\C^d)^{\otimes n}) = \mathcal{D}((\C^d)^{\otimes n})
    \end{equation*}
\end{proposition}

For \(\Tilde{\rho} \in \Tilde{\mathcal{D}}^{(2t)}((\C^d)^{\otimes n})\), while not a true density matrix, we can still view \(\Tilde{\rho}\) through it's reduced \(t\)-body moments, which are valid density matrices.

\subsection{Sum-of-Squares Certificates}

In this section, we give two main SOS certificates, dubbed the ``Star Bound'' and the ``Triangle Bound.'' We believe these certificates represent the notion of monogamy of entanglement from the Hamiltonian perspective. In doing this, we also prove the optimally of the degree-six SOS SDP\footnote{See, for example, \cite{navascues_convergent_2008,burgdorf_optimization_2016,takahashi_su2-symmetric_2023,watts_relaxations_2024}, for the relationship between non-commutative SOS hierarchies and SDPs.} on the star graph.

\begin{lemma}[The SOS Star Bound For The EPR Problem]\label{lemma_EPR_d_sos_star_bound}
    Let \(\Tilde{\rho} \in \Tilde{\calD}^{(6)}((\C^d)^{\otimes n})\) be a degree-six pseudo-density matrix. For any vertex \(a \in [n]\) and any subset \(S \subseteq [n] \setminus \{a\}\), we have that
    \begin{equation}
        \Tr\left(\tilde{\rho} \sum_{b \in S} \EPR^{a b}\right) \leq \frac{|S|+d-1}{d}
    \end{equation}
\end{lemma}

\begin{proof}
    It suffices to show that this holds for the star graph on \(n\) vertices with the root vertex labeled by 1. This is because partial traces of pseudo-density matrices are themselves pseudo-density matrices of the same degree. In particular, let \(H = \sum_{a = 2}^{n} \EPR^{1a}\), we then will show that that \(\Tr(H \tilde{\rho}) \leq \frac{n+d-2}{d}\) for any \(\Tilde{\rho} \in \Tilde{\calD}^{(6)}((\C^d)^{\otimes n})\). Note, \(H\) is defined to be hermitian (i.e., self-adjoint) and has degree \(\omega(H) = 2\).

    Letting \(M = CI - H\) with \(C\) to be determined later, we use \(\Tr\left(\Tilde{\rho} M^2\right) = \Tr\left(\Tilde{\rho} M^\dagger M\right) \geq 0\) to get an upper bound on \(\Tr\left(\Tilde{\rho} H\right)\). We start by considering \(H^2\), which gives the following:
    \begin{align}
        H^2 &= \sum_{a=2}^n \EPR^{1a} + \sum_{2 \leq a < b \leq n} \{\EPR^{1a}, \EPR^{1b}\} \\
        &= H + \frac{1}{d} \sum_{2 \leq a < b \leq n} \left(\EPR^{1a} + \EPR^{1b} - \frac{2(d-1)}{d} P^{1ab}\right) & \text{(\cref{lemma_anticom_of_epr_give_sum_and_proj})}\\
        &= H + \frac{1}{d} (n-2) H - \frac{2(d-1)}{d^2}\sum_{2 \leq a < b \leq n} P^{1ab}\\
        &= \left(\frac{n+d-2}{d}\right) H - \frac{2(d-1)}{d^2}\sum_{2 \leq a < b \leq n} P^{1ab}
    \end{align}
    Putting it all together:
    \begin{align}
        0 &\leq \Tr\left(\Tilde{\rho} M^2\right) \\
        &= C^2 - 2C \Tr\left(\Tilde{\rho} H \right) +  \left(\frac{n+d-2}{d}\right) \Tr\left(\Tilde{\rho}H\right) - \frac{2(d-1)}{d^2}\sum_{2 \leq a < b \leq n} \Tr\left(\Tilde{\rho} \left(P^{1ab}\right)^2\right) \\
        &\leq C^2 - 2C \Tr\left(\Tilde{\rho} H \right) +  \left(\frac{n+d-2}{d}\right) \Tr\left(\Tilde{\rho}H\right) \label{eq_sos_proof_positivity_prop_1}\\
        \Rightarrow \Tr\left(\Tilde{\rho}H\right) &\leq \frac{C^2}{2C - \left(\frac{n+d-2}{d}\right)}
    \end{align}
    Here, at \cref{eq_sos_proof_positivity_prop_1}, we use the fact that \(\omega(P^{1ab}) \leq 3\) and thus \(\Tr\left(\Tilde{\rho} \left(P^{1ab}\right)^2\right) \geq 0\) by the positivity constraint of degree-six pseudo-density matrices (\cref{def_pdm}, \cref{def_item_pdm_positivity}). Finally, setting \(C = \frac{n+d-2}{d}\) gives that \(\Tr\left(\Tilde{\rho}H\right) \leq \frac{n+d-2}{d}\), completing the proof.
\end{proof}

\begin{remark}
    The degree-six assumption comes from the fact that the projector, \(P\), in \cref{lemma_anticom_of_epr_give_sum_and_proj} acts non-trivially on at most three qudits, which gives an upper bound for its degree. If it is shown that \(P\) is actually degree-two, which is the case when \(d=2\), then this certificate requires only a degree-four assumption.
\end{remark}

We then get the immediate corollary of \cref{lemma_EPR_d_sos_star_bound} using \cref{lemma_equiv_of_QUG_EPR}. It can be easily verified that conjugating a pseudo-density matrix by local unitaries, as was done in \cref{eq_conj_by_loc_units_to_get_eqiv_state}, is still a pseudo-density matrix of the same degree and thus \cref{lemma_equiv_of_QUG_EPR} applies equally for pseudo-density matrices. This same reasoning was used in \cite{parekh_optimal_2022} to prove a similar result for the \(d=2\) case specifically, through the well-known \algprobm{Quantum Max-Cut} star bound for degree-4 pseudo-density matrices \cite{parekh_application_2021}.

\begin{proposition}[The SOS Star Bound]\label{prop_QUG_SOS_star_bound}
    Let \(\Tilde{\rho} \in \Tilde{\calD}^{(6)}((\C^d)^{\otimes n})\) be a degree-six pseudo-density matrix. For any sequence of unitaries, \((U_{ab} \in \SU(d)\ |\ 1 \leq a < b \leq n)\), let \((h_{ab})_{ab}\), for 
    \begin{align}
        h_{ab} \coloneq \left((I \otimes U_{ab})\ket{\EPR}\bra{\EPR}(I \otimes U^\dagger_{ab})\right) \in \calL((\C^d)^{\otimes 2})\text{,}
    \end{align}
    denote the corresponding sequence of projectors with \(h_{ab}^{ab} = h_{ab}^{ab} \otimes I^{[n] \setminus \{a,b\}} \in \calL((\C^d)^{\otimes n})\) denoting the projector applied to the \(ab\) subsystem. Then for any \(a \in [n]\) and any subset \(S \subseteq [n] \setminus \{a\}\), we have that
    \begin{equation}
        \Tr\left(\tilde{\rho} \sum_{b \in S} h_{ab}^{ab}\right) \leq \frac{|S|+d-1}{d} = \frac{|S|}{d} + \frac{d-1}{d}
    \end{equation}
\end{proposition}

We note that this bound is optimal in the following sense.

\begin{theorem}\label{thm_sos_epr_star_bound_is_opt}
    The maximum energy of the \algprobm{Maximal Entanglement} problem (and in particular the \(\EPR\) problem) on an unweighted star graph over \(n\) vertices is \(\frac{n+d-2}{d}\). Moreover, the degree-six SOS SDP is optimal for this graph. 
\end{theorem}

\begin{proof}[Proof (Sketch).]
    By \cref{lemma_equiv_of_QUG_EPR}, it suffices to find a witness state with energy \(\frac{n+d-2}{d}\) on the EPR problem over the star graph with \(n\) vertices/qudits and the root node being labeled by 1. One such state is

    \begin{align}
        \ket{\psi} &= \frac{1}{\sqrt{(n-1)(n+d-2)}} \sum_{k=1}^{d-1} \sum_{a = 2}^n \left(\ket{k}^1 \otimes \ket{k}^a \otimes \bigotimes_{b \in [n] \setminus \{1,a\}} \ket{d}^b\right) + \sqrt{\frac{n-1}{n+d-2}} \left(\bigotimes_{a \in [n]} \ket{d}^a\right) \\
        &= \sqrt{\frac{d}{(n-1)(n+d-2)}}\sum_{a = 2}^n \ket{\EPR}^{1a}\otimes \bigotimes_{b \in [n] \setminus \{1,a\}} \ket{d}^b\label{eq_sdncdinekjnasckljec}
    \end{align}
    where the superscripts give the ordering of the qudits. The optimality of the SOS SDP then follows by \cref{prop_QUG_SOS_star_bound}. Note, the choice of $\ket{d}$ above is arbitrary; any consistent single qudit state would work.
\end{proof}

\begin{figure}[t!]
    \centering
    \scalebox{0.6}{\begin{tikzpicture}[scale=2]
    \begin{scope}[local bounding box=graph]
        \node[circle, draw] (1) at (0,0) {$1$};
        \node[circle, draw] (2) at (0,1) {$2$};
        \node[circle, draw, fill=green20] (3) at (0.866,-0.5) {$3$};
        \node[circle, draw, fill=green20] (4) at (-0.866,-0.5) {$4$};
    
        \node[circle, draw] (1b) at (3,0) {$1$};
        \node[circle, draw, fill=green20] (2b) at (3,1) {$2$};
        \node[circle, draw] (3b) at (3.866,-0.5) {$3$};
        \node[circle, draw, fill=green20] (4b) at (3-0.866,-0.5) {$4$};
    
        \node[circle, draw] (1c) at (6,0) {$1$};
        \node[circle, draw, fill=green20] (2c) at (6,1) {$2$};
        \node[circle, draw, fill=green20] (3c) at (6.866,-0.5) {$3$};
        \node[circle, draw] (4c) at (6-0.866,-0.5) {$4$};
    
        \draw (3) -- (1) -- (2) node[midway,above=7,sloped,text=orange!70!black] {\small $\ket{\operatorname{EPR}}$} (1) -- (4);
        \draw (2b) -- (1b) -- (3b) node[midway,above=7,sloped,text=orange!70!black] {\small $\ket{\operatorname{EPR}}$} (1b) -- (4b);
        \draw (2c) -- (1c) -- (3c) (4c) -- (1c) node[midway,above=7,sloped,text=orange!70!black] {\small $\ket{\operatorname{EPR}}$};
    
    
        \node[text=green!50!black] at ($(3) + (0,0.32)$) {$\ket{d}$};
        \node[text=green!50!black] at ($(4) + (0,0.32)$) {$\ket{d}$};
    
        \node[text=green!50!black] at ($(2b) + (0.32,0)$) {$\ket{d}$};
        \node[text=green!50!black] at ($(4b) + (0,0.32)$) {$\ket{d}$};
    
        \node[text=green!50!black] at ($(2c) + (0.32,0)$) {$\ket{d}$};
        \node[text=green!50!black] at ($(3c) + (0,0.32)$) {$\ket{d}$};
    
        \node[text=black,scale=2] at (1.5,0.1853) {$\bm{+}$};
        \node[text=black,scale=2] at (4.5,0.1853) {$\bm{+}$};
    
        \begin{scope}[on background layer]
            \fill[orange30, rounded corners=9] 
            ($(1) + (-0.1575,-0.1575)$) -- ($(1) + (0.1575,-0.1575)$) -- 
            ($(2) + (0.1575,0.1575)$) -- ($(2) + (-0.1575,0.1575)$) -- cycle;
            \fill[orange30, rounded corners=9] 
            ($(1b) + (0.1525*-0.866,0.1525*0.5)+(0.1525*0.5,0.1525*0.866)$) -- 
            ($(1b) + (0.1525*-0.866,0.1525*0.5)+(-0.1525*0.5,-0.1525*0.866)$) -- 
            ($(3b) + (-0.1525*-0.866,-0.1525*0.5)+(-0.1525*0.5,-0.1525*0.866)$) -- 
            ($(3b) + (-0.1525*-0.866,-0.1525*0.5)+(0.1525*0.5,0.1525*0.866)$) -- cycle;
            \fill[orange30, rounded corners=9] 
            ($(1c) + (0.1525*0.866,0.1525*0.5)+(0.1525*-0.5,0.1525*0.866)$) -- 
            ($(1c) + (0.1525*0.866,0.1525*0.5)+(-0.1525*-0.5,-0.1525*0.866)$) -- 
            ($(4c) + (-0.1525*0.866,-0.1525*0.5)+(-0.1525*-0.5,-0.1525*0.866)$) -- 
            ($(4c) + (-0.1525*0.866,-0.1525*0.5)+(0.1525*-0.5,0.1525*0.866)$) -- cycle;
        \end{scope}
    \end{scope}

    \node[anchor=east, scale=2] at ($(graph.west)+(-0.3,0)$)
    {$\displaystyle\lvert\psi\rangle \propto$};
\end{tikzpicture}}
    \caption{A visual representation of the state in \cref{eq_sdncdinekjnasckljec} for the star graph on \(4\) vertices.%
    }
    \label{fig_state}
\end{figure}

We then move on to the triangle bound, for which we cannot use \cref{lemma_equiv_of_QUG_EPR}. We now use \cref{cor_anticom_of_UGC_give_sum_and_proj} directly. Note that we could have proven \cref{prop_QUG_SOS_star_bound} using this corollary, too.

\begin{proposition}[The SOS Triangle Bound]\label{prop_QUG_SOS_triangle_bound}
    In the situation of \cref{prop_QUG_SOS_star_bound}, take any three distinct indices \(a,b,c \in V\), then
    \begin{equation}
        \Tr\left(\tilde{\rho} \left(h_{ab}^{ab} + h_{bc}^{bc} + h_{ac}^{ac}\right)\right) \leq \frac{d+2}{d} = \frac{3}{d} + \frac{d-1}{d}
    \end{equation}
\end{proposition}

\begin{proof}
    Let \(H = h_{ab}^{ab} + h_{bc}^{bc} + h_{ac}^{ac}\) and \(\Tilde{\rho} \in \Tilde{\calD}^{(6)}((\C^d)^{\otimes n})\).
    Then, let \(M = CI - H\), for which we will again use that \(\Tr(\Tilde{\rho}M^2) \geq 0\). We start by considering \(H^2\), which gives the following.
    \begin{align}
        H^2 &= H + \{h_{ab}^{ab}, h_{bc}^{bc}\} + \{h_{ab}^{ab}, h_{ac}^{ac}\} + \{h_{bc}^{bc}, h_{ac}^{ac}\} \\
        &= H + \frac{2}{d} H - \frac{2(d-1)}{d^2} \sum_{i=1}^3 P^{abc}_i & \text{(\cref{cor_anticom_of_UGC_give_sum_and_proj})}
    \end{align}
    Putting it all together, we get the following.
    \begin{align}
        0 &\leq \Tr\left(\Tilde{\rho} M^2\right) \\
        &= C^2 - 2C \Tr\left(\Tilde{\rho} H \right) +  \left(\frac{d+2}{d}\right) \Tr\left(\Tilde{\rho}H\right) - \frac{2(d-1)}{d^2} \sum_{i=1}^3 \Tr\left(\tilde{\rho} P^{abc}_i\right) \\
        &\leq C^2 - 2C \Tr\left(\Tilde{\rho} H \right) +  \left(\frac{d+2}{d}\right) \Tr\left(\Tilde{\rho}H\right) & \text{(SOS Positivity)}\\
        \Rightarrow \Tr\left(\Tilde{\rho}H\right) &\leq \frac{C^2}{2C - \left(\frac{d+2}{d}\right)}
    \end{align}
    Setting \(C = \frac{d+2}{d}\) then gives that \(\Tr\left(\Tilde{\rho}H\right) \leq \frac{d+2}{d}\), completing the proof.
\end{proof}

\begin{remark}\label{remark_tri_bound_is_opt_EPR}
    While optimal for the EPR problem, if we restrict to the \algprobm{Quantum Max-Cut} Hamiltonian in the \(d=2\) case, for example, we can get a tighter bound of \(\frac{3}{2}\) as \(\sum_{i=1}^3 P^{abc}_i = H\), matching that of \cite{parekh_optimal_2022,wright_sosproofsemail_2023,lee_improved_2024}. We note, however, that our weaker bound is sufficient even for the analysis in \cite{lee_improved_2024}.
\end{remark}

\zacktodo{We can provide a lower bound matching this for the EPR problem. In particular it is the same idea as \cref{thm_sos_epr_star_bound_is_opt}.
\[C\left(\ket{\EPR}^{ab} \otimes \ket{d}^c + \ket{\EPR}^{ac} \otimes \ket{d}^b + \ket{\EPR}^{bc} \otimes \ket{d}^a\right)\]
where \(C\) is the value that makes it a unit vector (I think it should be \(\frac{d}{3(d+2)}\)).
I think this gives a very natural follow up question. Is the state the maximize the energy just the superposition over all matchings. 
}

\section{Analysis of The Matching-Based Algorithm}\label{section_matching_alg}

In this section, we analyze the performance of the matching-based algorithm originally proposed in \cite{lee_improved_2024} for \algprobm{Quantum Max-Cut}, which we adapt to the \algprobm{Maximal Entanglement} problem in \cref{algo_matching_algo_for_QUG}.

\begin{algorithm}[\algprobm{Maximal Entanglement} Problem Matching Algorithm]\label{algo_matching_algo_for_QUG} \emph{Input:} Graph, \(G=(V,E,w)\), and unitaries \((U_e \in \SU(d)\ |\ e \in E)\).
    \begin{enumerate}
        \item\label{step_1_algo_matching_algo_for_QUG} Find the maximum matching of \(G\),  denoted by \(m : E \ra \{0,1\}\) (e.g., using Blossom algorithm \cite{edmonds_maximum_1965}). 
        \item \emph{Output:} \(\displaystyle\rho\coloneq \!\!\!\!\!\! \bigotimes_{\substack{(a,b) \in E:\\ m((a,b)) = 1}}\!\!\!\!\left((I \otimes U_{ab})\ket{\EPR}\bra{\EPR}(I \otimes U^\dagger_{ab})\right)^{ab} \otimes \!\!\!\!\!\!\!\!\bigotimes_{\substack{c \in V:\\\forall d \in V : m((c,d))=0}} \!\!\!\!\!\!\!\! I^c/d\)
    \end{enumerate}
\end{algorithm}

The key insights needed for this analysis are from \cite{edmonds_maximum_1965,lee_improved_2024}. We state the needed lemma below, where \(\OPT_{\mathsc{Match}}(G) = \E_{e \sim E} m(e) = \frac{1}{W}\sum_{e \in E} m(e) w_e\) denotes the normalized/expected maximum matching of a graph, \(G\), and \(W \coloneq \sum_{e \in E} w_e\).

\begin{lemma}[\cite{edmonds_maximum_1965,lee_improved_2024}]\label{lemma_lee_sdp_matching_lemma}
    Given a weighted graph, \(G=(V,E,w)\), define \(N(a) = \{b \in V\ |\ (a,b) \in E\}\) and \(E(S) = \{(a,b) \in E\ |\ \{a,b\} \subseteq S\}\) for \(S \subseteq V\). If \((x_e)_{e \in E}\) is a sequence of scalars such that
    \begin{enumerate}
        \item \(\forall e \in E : x_e \geq 0\),
        \item \(\forall a \in V : \sum\limits_{b \in N(a)} x_{ab} \leq 1\), and
        \item\label{lemma_item_tri_bound_for_matching} \(\forall \{a,b,c\} \subseteq V : \sum\limits_{e \in E(\{a,b,c\})} x_{e} \leq 1\)
    \end{enumerate}
    then
    \begin{equation}
        \frac{4}{5}\E_{e \sim E} \left[x_e\right] = \frac{4}{5 W} \sum_{e \in E} w_e x_e \leq \OPT_{\mathsc{Match}}(G)\text{.}
    \end{equation}
\end{lemma}

Similar to \cite{lee_improved_2024}, we will show that any degree-six pseudo-density matrix over a \algprobm{Maximal Entanglement} instance admits such a sequence of scalars, which is the key step needed to prove the forthcoming \cref{thm_upper_bound_on_QUG_energy}. To this end, we make the following definitions.

\begin{definition}\label{def_xy_values_for_matching}
    Let \(\Tilde{\rho} \in \Tilde{\calD}^{(6)}((\C^d)^{\otimes n})\) be a degree-six pseudo-density matrix. For any sequence of unitaries, \((U_{ab} \in \SU(d)\ |\ 1 \leq a < b \leq n)\), let \((h_{ab})_{ab}\), for 
    \begin{align}
        h_{ab} = \left((I \otimes U_{ab})\ket{\EPR}\bra{\EPR}(I \otimes U^\dagger_{ab})\right) \in \calL((\C^d)^{\otimes 2})\text{,}
    \end{align}
    denote the corresponding sequence of projectors with \(h_{ab}^{ab} = h_{ab}^{ab} \otimes I^{[n] \setminus \{a,b\}} \in \calL((\C^d)^{\otimes n})\) denoting the projector applied to the \(ab\) subsystem. We additionally define the following scalars:
    \begin{equation}
        x_{ab} = \Tr(\Tilde{\rho} h_{ab}^{ab}),\ \ \ \ y_{ab} = x_{ab} - \frac{1}{d},\ \ \ \ y^+_{ab} = \max(0, y_{ab})
    \end{equation}
\end{definition}

Given a \algprobm{Maximal Entanglement} instance over a weighted graph, \(G=(V,E,w)\), we can define the full sequence of unitaries and scalars, as in \cref{def_xy_values_for_matching}, by letting \(U_{ab} = I\) when \((a,b) \notin E\). We note that the significance of this choice is entirely unimportant and is only done for simplicity of the expressions. With this notation in hand, we then give the following two lemmas.

\begin{lemma}\label{lemma_QUG_SOS_star_bound_as_fractional_matching}
    In the situation of \cref{def_xy_values_for_matching}, for any \(a \in [n]\) and \(S \subseteq [n] \setminus \{a\}\), we have that
    \begin{equation}
        \sum_{b \in S} y_{ab}^+ \leq \frac{d-1}{d}\text{.}
    \end{equation}
\end{lemma}

\begin{proof}
    We make the observation that \(\sum_{b \in S} y_{ab}^+ = \sum_{b \in S : y_{ab} \geq 0} y_{ab}\). Then applying \cref{prop_QUG_SOS_star_bound} to the second summation, and in particular the subset \(\{b \in S : y_{ab} \geq 0\}\), proves the claim.
\end{proof}

\begin{lemma}\label{lemma_QUG_SOS_triangle_bound_as_fractional_matching_const_2}
    In the situation of \cref{def_xy_values_for_matching}, for any three distinct values \(\{a,b,c\} \subseteq [n]\), we have that
    \begin{equation}
        y_{ab}^+ + y_{bc}^+ + y_{ac}^+ \leq \frac{d-1}{d}\text{.}
    \end{equation}
\end{lemma}

\begin{proof}
    We split this into three cases depending on the number of non-zero \(y^+_{uv} > 0\) values. In particular, let \(t = |\{(u,v) \in \{(a,b),(b,c),(a,c)\}\ |\ y^+_{uv} > 0\}|\). First, for \(t=1\), we have that, \(-\frac{1}{d} \leq y_{uv} \leq \frac{d-1}{d}\) by the fact that the edge Hamiltonian is a projector. Otherwise, if \(t=2\), this follows directly by \cref{lemma_QUG_SOS_star_bound_as_fractional_matching}. Finally, for the \(t=3\) case, we use \cref{prop_QUG_SOS_triangle_bound} which tells us that
    
    \begin{equation}
        y^+_{ab} + y^+_{bc} + y^+_{ac} = y_{ab} + y_{bc} + y_{ac} = x_{ab} + x_{bc} + x_{ac} - \frac{3}{d} \leq \frac{d+2}{d} - \frac{3}{d} = \frac{d-1}{d}\text{.}
    \end{equation}
\end{proof}

Thus, by \cref{lemma_lee_sdp_matching_lemma,lemma_QUG_SOS_star_bound_as_fractional_matching,lemma_QUG_SOS_triangle_bound_as_fractional_matching_const_2} we have that for any \algprobm{Maximal Entanglement} instance and any degree-six pseudo-density matrix, the scalars defined in \cref{def_xy_values_for_matching} (in particular, \((\frac{d}{d-1}y^+_e)_e\)) satisfy the premise in \cref{lemma_lee_sdp_matching_lemma}, and thus we have that \(\frac{4d}{5(d-1)}\E_{e \sim E} \left[y^+_e\right] \leq \OPT_{\mathsc{Match}}(G)\).

\begin{theorem}\label{thm_upper_bound_on_QUG_energy}
    Let \(\langle G = (V,E,w), (U_e)_e \rangle\) be an instance of the \algprobm{Maximal Entanglement} problem with normalized Hamiltonian \(H\). Then for any degree-six pseudo-density matrix \(\Tilde{\rho} \in \Tilde{\calD}^{(6)}((\C^d)^{\otimes n})\), we have that
    \begin{equation}
        \Tr(H \Tilde{\rho}) \leq \frac{1}{d} + \frac{5(d-1)}{4d}\OPT_{\mathsc{Match}}(G)\text{.}
    \end{equation}
\end{theorem}

\begin{proof}
    As mentioned above, this follows from \cref{lemma_lee_sdp_matching_lemma,lemma_QUG_SOS_star_bound_as_fractional_matching,lemma_QUG_SOS_triangle_bound_as_fractional_matching_const_2}.
    \begin{align}
        \Tr(H \Tilde{\rho})
        &= \E_{(a,b) \sim E} \left[\frac{1}{d} + y_{ab}\right] \\
        &\leq \frac{1}{d} + \E_{(a,b) \sim E} \left[y^+_{ab}\right] \label{eq_max_with_zero_greater_than_y} \\
        &\leq \frac{1}{d} + \frac{5(d-1)}{4d}\OPT_{\mathsc{Match}}(G) 
    \end{align}
    Note, to get \cref{eq_max_with_zero_greater_than_y}, we use the fact that \(y_e \leq y^+_e\).
\end{proof}


In particular, over \(D\)-regular graphs we can use the fact that \(\OPT_{\mathsc{Match}}(G) \leq O(\frac{1}{D})\), which gives the bound \(\Tr(H \Tilde{\rho}) \leq \frac{1}{d} + O(\frac{1}{D})\). To explore this further, we parameterize the approximation ratio in terms of the expected maximum matching of the interaction graph.

We next consider bounding the approximation ratio of \(\cref{algo_matching_algo_for_QUG}\). It will be useful to bound this value over different classes of graphs such as graphs with fixed maximum matching value, \(M\), and \(D\)-regular graphs. In both cases we implicitly assume that there are infinite families of such graphs in order for the asymptotics to be well-defined.
We let \(\alpha_d\) denote the approximation ratio over all graphs and local-dimension, \(d\), of the Hamiltonian, we let \(\alpha_d(M)\) denote the approximation ratio over all graphs with maximum matching \(\OPT_{\mathsc{Match}}(G) = M\), and lastly, we let \(\alpha_d(D)\) denote the approximation ratio over all (unweighted) \(D\)-regular graphs. We seek to lower bound these constants.

\begin{theorem}\label{thm_main_genral_graph}
    Let \(\langle G = (V,E,w), (U_e)_e \rangle\) be an instance of the \algprobm{Maximal Entanglement} and let \(M \in (0,1]\), then \cref{algo_matching_algo_for_QUG} has an approximation ratio of \(\alpha_d(M) \geq \frac{4}{5}\frac{\left(d^2-1\right) M+1}{d ((d-1) M+\frac{4}{5})} = \frac{1}{d} + \Omega_{M \to 0}(M)\) on the energy of the Hamiltonian over all graphs, \(G\), with \(\OPT_{\mathsc{Match}}(G) = M\).
\end{theorem}

\begin{proof}
    We let \(\rho_* \in \calD((\C^d)^{\otimes n})\) be the most excited state of the problem Hamiltonian of the ME instance, \(H\), and \(\rho \in \calD((\C^d)^{\otimes n})\) be the output solution of \cref{algo_matching_algo_for_QUG}. Let \(m: E \to \{0,1\}\) be the corresponding matching from \cref{step_1_algo_matching_algo_for_QUG} of \cref{algo_matching_algo_for_QUG}. We can observe that the energy of \(\rho\) is then given by the following:
    \begin{align}
        \Tr(\rho H) &= \E_{(a,b) \sim E}\left[\frac{1}{d^2} + \frac{d^2 - 1}{d^2}m((a,b))\right] \\
        &= \frac{1}{d^2} + \frac{d^2 - 1}{d^2}\OPT_{\mathsc{Match}}(G) \label{eq_ssdasdasdgjyd}
    \end{align}
    Indeed, when \(m((a,b)) = 1\) then the energy along the edge \((a,b)\) is \(\Tr(\rho h_{ab}^{ab}) =  1\) and when \(m((a,b)) = 0\) we observe that the reduced density matrix is nothing but the maximally mixed state, \(\rho^{ab} = \frac{1}{d^2}I\), and thus \(\Tr(\rho h_{ab}^{ab}) = \frac{1}{d^2}\).
    
    We then give a lower bound on the approximation ratio using \cref{thm_upper_bound_on_QUG_energy}. Recall that \(\alpha_d(M)\) is a constant, depending only on \(d\) and \(M\), such that
    \begin{equation}\label{eq_sjdcancsoiudfhsdf}
        \Tr(\rho H) \geq \alpha_d(M) \left(\frac{5(d-1)}{4d}\OPT_{\mathsc{Match}}(G) + \frac{1}{d}\right) \geq \alpha_d(M) \Tr(\rho_* H)
    \end{equation}
    for all \(G\) with \(\OPT_{\mathsc{Match}}(G) = M\). Solving for \(\alpha_d(M)\), we get the following, where the infimum is taken over all graphs with maximum matching, \(\OPT_{\mathsc{Match}}(G) = M\).
    
    \begin{align}
        \alpha_d(M) &= \inf_G\left(\frac{\Tr(\rho H)}{\Tr(\rho_* H)}\right) \\
        &\geq \inf_G\left(\frac{\frac{1}{d^2} + \frac{d^2-1}{d^2} \OPT_{\mathsc{Match}}(G)}{\frac{1}{d} + \frac{5(d-1)}{4d}\OPT_{\mathsc{Match}}(G)}\right) & \text{(Apply \cref{eq_ssdasdasdgjyd,eq_sjdcancsoiudfhsdf})} \\
        &= \frac{4}{5}\frac{\left(d^2-1\right) M+1}{d ((d-1) M+\frac{4}{5})} \label{eq_remove_inf_G} \\
        & = \frac{1}{d} + \Omega_{M\ra 0}(M) \label{eq_big_O_in_matching}
    \end{align}
    To get \cref{eq_remove_inf_G} we make the observation that no property of \(G\), other than \(\OPT_{\mathsc{Match}}(G) = M\), is used, and thus we can drop the infimum. Then, the asymptotics in \cref{eq_big_O_in_matching} follow by observing that the following limit is non-zero.
    \begin{equation}
        \limsup_{M \ra 0}\left(\left(\frac{4}{5}\frac{\left(d^2-1\right) M+1}{d ((d-1) M+\frac{4}{5})} - \frac{1}{d}\right)\bigg/M\right) = \frac{4d^{2}-5d+1}{4d} > 0
    \end{equation}
\end{proof}

\begin{remark}
    To better understand the relationship between the approximation ratio and the maximum matching we considered the asymptotics of \(\alpha_d(M)\) parameterized by \(M\) in \cref{thm_main_genral_graph}. Intuitively, we expect \cref{algo_matching_algo_for_QUG} to perform worse when the maximum matching is poor, hence why we consider the asymptotics as \(M \to 0\). When not parameterized by the maximum matching we get an approximation ratio of \(\alpha_d  \geq \frac{1}{d}\). This follows by observing that the infimum of \(\alpha_d(M)\) over all graphs reduces to taking the limit as \(M \ra 0\). Indeed, the expression \(\frac{4}{5}\frac{\left(d^2-1\right) M+1}{d \left(\left(d-1\right) M+\frac{4}{5}\right)}\) is minimized when \(M \to 0\).
    
    On the other hand, for any fixed constant \(M > 0\), taking the limit as \(d \ra \infty\) gives a bound on the approximation ratio being \(\lim_{d \to \infty}(\alpha_d(M)) \geq \frac{4}{5}\).
    We, however, expect \cref{algo_matching_algo_for_QUG} to be optimal over graphs with a maximum matching of \(M = 1\). This discrepancy hints that our analysis is likely not tight and better certificates can probably be found. Indeed, better certificates would improve on the \(\frac{4}{5}\) constant in \cref{lemma_lee_sdp_matching_lemma}. See \cref{section_the_qubit_case} for a more in-depth analysis on this matter in the qubit case.
\end{remark}

Lastly, we use \cref{thm_main_genral_graph} to consider the case of bounded degree graphs. We believe this to be an interesting case because as was shown in \cite{brandao_product-state_2016,parekh_optimal_2022} product state algorithms/ansatz work well in the high degree setting or bounded minimum degree setting. Therefore, understanding how well an algorithm with no global entanglement performs in the bounded degree setting can shed light on approximation algorithms as a whole.

\begin{theorem}\label{thm_main_bounded_min_max_deg_graph}
    For an instance of the \algprobm{Maximal Entanglement} problem with unweighted \(D\)-regular graph, \(G = (V,E)\), \cref{algo_matching_algo_for_QUG} has an approximation ratio of \(\alpha_d(D) \geq \frac{4}{5}\frac{\left(d^2 \left(D+2\right)+D \left(D+2\right)-2\right)}{d \left(D+3\right) \left(d+\frac{4}{5} D-1\right)} = \frac{1}{d} + \Omega(\frac{1}{D})\) on the energy of the Hamiltonian.
\end{theorem}

\begin{proof}
    Before considering the case of \(D\)-regular graphs, let \((D_{\min}, D_{\max})\) be the minimum and maximum degrees of the graph. We can then bound the expected maximum matching in the following way, \(\frac{1}{D_{\min}} \geq \OPT_{\mathsc{Match}}(G) \geq \frac{(D_{\max}+2)}{D_{\max}(D_{\max}+3)}\) by \cite{han_tight_2012}. Plugging this into the result of \cref{thm_main_genral_graph} gives
    \begin{equation}\label{eq_bound_min_max_deg}
        \alpha_d \geq \frac{4}{5}\frac{D_{\min} \left(\frac{ \left(d^2-1\right) (D_{\max}+2)}{D_{\max} (D_{\max}+3)}+1\right)}{d (d+\frac{4}{5} D_{\min}-1)}\text{.}
    \end{equation}
    In the \(D\)-regular case, this becomes
    \begin{equation}
        \alpha_d \geq \frac{4}{5}\left(\frac{d^{2}\left(D+2\right)+D\left(D+2\right)-2}{d\left(D+3\right)\left(d+\frac{4}{5}D-1\right)}\right)\text{.}
    \end{equation}
    We can then rewrite the above expression by pulling out a \(\frac{1}{d}\) term, which gives the following:
    \begin{align}
        \alpha_d &\geq \frac{1}{d} + \frac{(d-1) (4 d (D+2)-D-7)}{d (D+3) (5 d+4 D-5)}
    \end{align}
    We then consider the following limit:
    \begin{align}
        \limsup_{D\ra\infty}\left(\frac{(d-1) (4 d (D+2)-D-7)}{d (D+3) (5 d+4 D-5)} \cdot D\right) &= \frac{4d^{2}-5d+1}{4d}
    \end{align}
    And thus we have that \(\alpha_d(D) \geq \frac{1}{d} + \Omega(\frac{1}{D})\).
\end{proof}

\begin{remark}
    We note that the asymptotic applies even to graphs with bounded degree, i.e., when \(D_{\min} = 1\) in \cref{eq_bound_min_max_deg}, but we get a worse absolute constant, \(\frac{(d-1) (4 d-1)}{d (5 d-1)} \ra \frac{4}{5}\). This is likely due to the bound on the expected matching not being tight in terms of \(D_{\min}\) and \(D_{\max}\) rather than this being a harder case.
\end{remark}

Alone, this result is hard to interpret. However, we get the following highly non-trivial corollary. 

\begin{corollary}\label{cor_main_bounded_deg_5}
    In the situation of \cref{thm_main_bounded_min_max_deg_graph} and for all \(d \geq 2\), we have an approximation guarantee of strictly greater than \(\frac{1}{2}\) over \(D\)-regular graphs with \(D \leq 5\).
\end{corollary}

\begin{proof}
    We note that in the \(D=1\) case all such graphs are matchings and thus the algorithm would find the (unique) maximum energy state.
    For \(D \geq 2\), taking the derivative of the expression in \cref{thm_main_bounded_min_max_deg_graph} and solving for zero gives us that it is at its minimum when \(d = \frac{5\left(D\left(D+2\right)-2\right)}{\left(D+2\right)\left(4D-5\right)} + \sqrt{\frac{D\left(D\left(D+2\right)-2\right)\left(D\left(16D+17\right) - 5\right)}{\left(5-4 D\right)^2\left(2+D\right)^2}}\).\footnote{We calculated this using Mathematica.} Finally, we use this to get the following bounds on the approximation ratio with respect to \(D\): \(\alpha_d(2) \geq 0.6056, \alpha_d(3) \geq 0.5736, \alpha_d(4) \geq 0.5435, \alpha_d(5) \geq 0.5174\).
\end{proof}

We note the following two easy-to-verify results that can be used to compare to our result. We note that the approximation ratio achieved by the maximally mixed state, \(\rho^{\mathsc{mixed}} \coloneq \frac{I}{d^{|V|}}\), gives a stand-in for the random assignment bound.

\begin{proposition}\label{prop_prod_and_max_mixed_approx} Let \(\langle G = (V,E,w),(U_e)_{e\in E} \rangle\) be any instance of the \algprobm{Maximal Entanglement} problem with Hamiltonian \(H\).
    \begin{enumerate}
        \item\label{prop_item_prod_and_max_mixed_approx_1} Any product state \(\ket{\psi} = \bigotimes_{a =1}^n \ket{\psi_a} \in (\C^d)^{\otimes V}\) has energy \(\bra{\psi} H \ket{\psi} \leq \frac{1}{d}\). This is tight on the edge graph, giving an upper bound of \(\frac{1}{d}\) on the approximation ratio with product states. In fact, there always exists a product state solution with energy within \(\frac{1}{d}\) of the optimal \cite{gharibian_approximation_2012}.
        \item The maximally mixed state has energy \(\Tr(\rho^{\mathsc{mixed}} H) = \frac{1}{d^2}\). As an approximation algorithm, this is tight on the edge graph, giving an approximation ratio of \(\frac{1}{d^2}\). 
    \end{enumerate}
\end{proposition}

\begin{remark}\label{remark_sdasdhucakjSAEkub}
    We note that in the limit as the graph degree \(D \ra \infty\), the optimal energy is bounded above by \(\frac{1}{d}\) (as per \cref{thm_upper_bound_on_QUG_energy}) and thus the maximally mixed state achieves an approximation ratio of \(\frac{1}{d}\) under this limit. When parameterized by \(D\), we can observe that \(\Tr(\rho^{\mathsc{mixed}} H) \geq \left(\frac{1}{d} - \Theta(\frac{1}{D})\right)\Tr(\rho_* H)\).
\end{remark}

With this, we claim that \cref{algo_matching_algo_for_QUG} beats random assignment for any finite degree. Indeed, from \cref{eq_ssdasdasdgjyd} it is evident that the state output by \cref{algo_matching_algo_for_QUG} has strictly more energy than the maximally mixed state.

\section{\textsc{Quantum Max-Cut}, The EPR Problem, and The Qubit Case}\label{section_the_qubit_case}

In this section, we consider three special cases for which we can achieve better results through the addition of algorithms that return product state solutions. In all cases, we run both the basic matching algorithm from \cref{algo_matching_algo_for_QUG} and a product state rounding algorithm, and then return the state with the greater energy.

The first special case is the EPR problem, which has an optimal product state solution \(\ket{1}^{\otimes n}\) for any interaction graph \(G = (V,E)\). This product state has expected energy \(\E_{e \in E} \Tr\left(\ket{1}^{\otimes 2} \bra{1}^{\otimes 2} \EPR_d\right) = 1/d\).

\begin{theorem}\label{thm_main_epr_genral_graph}
    For any instance of the \(d\)-dimensional qudit EPR problem, the above-described algorithm has an approximation guarantee of \(\frac{4(d+1)}{9d-1}\) on the energy of the Hamiltonian. 
\end{theorem}

\begin{proof}
    Let \(G=(V,E,w)\) be an interaction graph for the EPR problem with Hamiltonian \(H\), let \(m: E \ra \{0,1\}\) be the maximal matching for the graph, and \(\rho\) the outputted state, by the matching based algorithm from \cref{algo_matching_algo_for_QUG}. Let \(\rho_* \in \calD((\C^d)^{\otimes n})\) be the most excited state of the problem Hamiltonian and let \(x_e\) be defined as in \cref{def_xy_values_for_matching} for \(\rho_*\). By \cref{lemma_lee_sdp_matching_lemma,lemma_QUG_SOS_star_bound_as_fractional_matching,lemma_QUG_SOS_triangle_bound_as_fractional_matching_const_2} we have that \(\E_{e \in E} m(e) \geq \frac{4d}{5(d-1)}\E_{e \in E} \max\left(0,x_e - \frac{1}{d}\right)\) and thus
    \begin{equation}
        \Tr(\rho H) = \E_{e \in E} \left(\frac{1}{d^2} + \frac{(d^2-1)}{d^2} m(e)\right) \geq \E_{e \in E} \left(\frac{1}{d^2} + \frac{4(d+1)}{5d} \max\left(0,x_e - \frac{1}{d}\right)\right)\text{.}
    \end{equation}

    Next, we note that the state \(\ket{1}^{\otimes n}\) achieves energy \(\Tr(\ket{1}^{\otimes n} \bra{1}^{\otimes n} H) = \frac{1}{d}\).
    We then bound the approximation ratio, \(\beta_d\), by considering the ``worst case edge'' and using the fact that \(\max(x,y) \geq px + (1-p)y\) for any $p \in [0,1]$.
    \begin{align}
        \beta_d &=  \inf_G \frac{\max\left\{\Tr(\ket{1}^{\otimes n} \bra{1}^{\otimes n} H), \Tr(\rho H)\right\}}{\Tr(\rho_* H)} \\
        &\geq \max_{p \in [0,1]} \min_{x \in [0,1]} \left(p\frac{1}{d x} + (1-p)\frac{\frac{1}{d^2} + \frac{4(d+1)}{5d} \max\left(0,x - \frac{1}{d}\right)}{x}\right) \\
        &= \frac{4(d+1)}{9d-1}
    \end{align}
    This is achieved when \(p = \frac{4d-1}{9d-1}\) and \(x = \frac{1}{d}\).
\end{proof}

In fact, we can do better than this in the qubit (\(d=2\)) case and for the \algprobm{Quantum Max-Cut (QMC)} Hamiltonian \cite{gharibian_almost_2019}. We have the following generalization of \cref{lemma_QUG_SOS_triangle_bound_as_fractional_matching_const_2}.

\begin{lemma}[The \(K_n\) Matching Bound]\label{lemma_Kn_matching_bound}
    In the situation of \cref{def_xy_values_for_matching}, with \(\Tilde{\rho} \in \Tilde{\calD}^{(2t)}((\C^d)^{\otimes n})\) being a degree-\(2t\) pseudo-density matrix, if for every graph, \(G=([n],E)\), for odd \(n\), we have
    \begin{equation}
        \sum_{(a,b) \in E(G)} x_{ab} = \Tr\left(\Tilde{\rho} \sum_{(a,b) \in E(G)} h^{ab}_{ab}\right) \leq \frac{(d-1)(n-1)}{2d} + \frac{|E|}{d}
    \end{equation}
    then we have that \(\sum_{1 \leq a < b \leq n} y^+_{ab} \leq \frac{(d-1)(n-1)}{2d}\).
\end{lemma}

\begin{proof}
    For any valid \(\Tilde{\rho} \in \Tilde{\calD}^{(2t)}((\C^d)^{\otimes n})\), consider the subgraph, \(G=([n],E)\), of the complete graph, \(K_n\), with the edge \((a,b) \in G\) precisely when \(y^+_{ab} > 0\). It then follows from the assumption that
    
    \begin{equation}
        \sum_{1 \leq a < b \leq n} y^+_{ab} = \sum_{(a,b) \in E(G)} y^+_{ab} = \sum_{(a,b) \in E(G)} y_{ab} = \sum_{(a,b) \in E(G)} x_{ab} - \frac{|E|}{d} \leq \frac{(d-1)(n-1)}{2d}
    \end{equation}
    completing the proof.
\end{proof}

Next, we can verify that the assumption of \cref{lemma_Kn_matching_bound} holds for the EPR problem when \(d=2\) and the QMC problem for \(n=5\). In particular, we calculate the maximal energy of the EPR/QMC Hamiltonian's over every graph on five vertices, up to isomorphism. For which, there are \(33\) non-empty five vertex graphs. For eight of these graphs, the required bounds follow by \cref{lemma_QUG_SOS_star_bound_as_fractional_matching,lemma_QUG_SOS_triangle_bound_as_fractional_matching_const_2} already.

\begin{lemma}[Computational]\label{lemma_K5_matching_bound_computational}
    Let \(\rho \in \calD((\C^2)^{\otimes 5})\) be a true density matrix on five qudits. Let \(\EPR^{ab} = \ket{\EPR}\bra{\EPR}^{ab} \otimes I^{[n] \setminus \{a,b\}} \in \calL((\C^2)^{\otimes 5})\) denote the projector onto the \(\EPR\) state applied to the \(a\) and \(b\) systems and identity on all other systems.
    Additionally, let \(\operatorname{QMC}^{ab} = \ket{\Psi^-}\bra{\Psi^-}^{ab} \otimes I^{[n] \setminus \{a,b\}} \in \calL((\C^2)^{\otimes 5})\) denote the projector onto the singlet state, \(\ket{\Psi^-} \coloneq \frac{1}{\sqrt{2}} \left(\ket{01} - \ket{10}\right)\), applied to the \(a\) and \(b\) systems and identity on all other systems.
    We have the following for every graph, \(G\), on \(5\) vertices:
    \begin{equation}
        \Tr\left(\rho \sum_{(a,b) \in E(G)} \EPR^{ab}\right) \leq \frac{(n-1)}{4} + \frac{|E|}{2}
    \end{equation}
    \begin{equation}
        \Tr\left(\rho \sum_{(a,b) \in E(G)} \operatorname{QMC}^{ab}\right) \leq \frac{(n-1)}{4} + \frac{|E|}{2}
    \end{equation}
\end{lemma}

\begin{proof}
    This is verified by solving for the maximum eigenvalue for each case, computationally. 
\end{proof}

We then get the following strengthening of \cref{thm_main_epr_genral_graph}.

\begin{theorem}\label{thm_main_epr_genral_graph_2}
    For any instance of the EPR problem (with \(d=2\)), the algorithm in \cref{thm_main_epr_genral_graph} has an approximation guarantee of \(\frac{18}{25} = 0.72\) on the energy of the Hamiltonian. 
\end{theorem}

\begin{proof} If in the assumptions for \cref{lemma_lee_sdp_matching_lemma} we also have that \(\forall S \subseteq V : |S| = 5 \ra \sum_{e \in E(S)} x_e \leq 2\), then we have that \(\frac{6}{7}\E_{e\in E} [x_e] \leq \OPT_\mathsc{Match}(G)\) \cite{lee_improved_2024}. \cref{lemma_K5_matching_bound_computational} gives \(\forall S \subseteq V : |S| = 5 \ra \sum_{e \in E(S)} 2 y^+_e \leq 2\), as required, and thus we have that \(\frac{12}{7}\E_{e\in E} [x_e] \leq \OPT_\mathsc{Match}(G)\).
The rest of the proof is nearly identical to that of \cref{thm_main_epr_genral_graph} and the bound is achieved by \(p=\frac{11}{25}\) and \(x=\frac{1}{2}\).
\end{proof}


\begin{theorem}\label{thm_bigger_number_for_qmc}
    For any instance of the QMC problem, the algorithm from \cite{lee_improved_2024} achieves an approximation guarantee of 0.599.
\end{theorem}

\begin{proof}
    The proof is near identical to that of \cite[Theorem 11]{lee_improved_2024}, but we use \cref{lemma_K5_matching_bound_computational} and the improved version of \cref{lemma_lee_sdp_matching_lemma} from the proof of \cref{thm_main_epr_genral_graph_2} to get \(\frac{12}{7}\E_{e\in E} [x_e] \leq \OPT_\mathsc{Match}(G)\). The bound is then achieved by \(p = 0.697\) and \(x = -0.4353\).
\end{proof}

\subsection{The Qubit Case}

Next, we consider the general Maximal Entanglement (ME) problem over qubits. Here, we use the SDP based product state rounding algorithm from \cite{parekh_beating_2021} that performs well in bounded minimum degree instance \cite{parekh_optimal_2022}. We note that the SDP considered in \cite{parekh_beating_2021} gives a degree-2 pseudo-density matrix such that all 2-moments are valid density matrices. For an instance of the ME problem, let \(\Tilde{\rho}\) be the pseudo-density matrix resulting from the optimal SDP variables. Additionally, for some vector of i.i.d. random Gaussian's, \(\mathbf{r} \in \calN(0,I)\),
let \((s_v \in S^{2}\ |\ v \in V)\) be the rounded Bloch vector representation of the result of \cite[Algorithm 9]{parekh_beating_2021} with corresponding single qubit states \((\rho_v)_v\). We then state the main lemma that we will need.

\begin{lemma}[\cite{parekh_beating_2021}]\label{lemma_from_PT21b_1}
    For a ME instance with edge Hamiltonians, \((h_{uv})_{uv}\), and optimal pseudo-density matrix \(\Tilde{\rho}\) and rounded product state solution \((\rho_v)_v\) as defined above, one has that, for any edge \((u,v) \in E\),
    \begin{equation}\label{eq_for_approx_ratio_numerator}
        \E_{\mathbf{r}}\left[\Tr(h_{uv}\rho_u \otimes \rho_v)\right] = \frac{1}{4}\left(1 + \E_{(z,z')}\left[\frac{\braket{[p,q,r]^\sfT,[z_1 z_1', z_2 z_2', z_3 z_3']^\sfT}}{\sqrt{(z_1^2 + z_2^2 + z_3^2)((z_1')^2 + (z_2')^2 + (z_3')^2)}}\right]\right)
    \end{equation}
    \begin{equation}\label{eq_for_approx_ratio_denominator}
        \E_{\mathbf{r}}\left[\Tr(h_{uv}\Tilde{\rho})\right] = \frac{1}{4}\left(1 + \braket{[p,q,r]^\sfT,[a,b,c]^\sfT}\right)
    \end{equation}
    with \((z,z') \sim \calN(0,\Sigma_{6}(\operatorname{diag}(a,b,c)))\) for some constants \((a,b,c,p,q,r) \in \calS \times \calS\). Here, we use \(\calS \coloneq \operatorname{conv}\left\{(-1,-1,-1), (-1,1,1),(1,-1,1),(1,1,-1)\right\}\) and \(\Sigma_{2n}(A) \coloneq I_{2n} + \left[\begin{smallmatrix}
        0 & 1 \\ 1 & 0
    \end{smallmatrix}\right] \otimes A\) for \(A \in M_{n}(\R)\).
\end{lemma}

We roughly follow the steps of \cite{parekh_beating_2021,parekh_optimal_2022} and analyze the approximation ratio using a high order Hermite expansion and numerical optimization techniques. In particular, we use the following lemmas.

\begin{lemma}[\cite{parekh_beating_2021}]\label{lemma_from_PT21b_2}
    In the context of \cref{lemma_from_PT21b_1} we have that
    \begin{equation}
        \cref{eq_for_approx_ratio_numerator} = \frac{1}{4}\left(1 + \sum_{\substack{i,j,k\\j \leq k}} \hat{f}^2_{i,jk}\left(p\, u_{i,jk}(a,b,c) + q\, u_{i,jk}(b,a,c) + r\, u_{i,jk}(c,a,b)\right)\right)\text{,}
    \end{equation}
    where we define the following functions, using \(p = (i + j + k - 1)/2\),
    \begin{equation}
        \hat{f}_{i,jk} = \begin{cases}
            2 \sqrt{\frac{2}{\pi}}\frac{(-1)^p \sqrt{i! j! k!}}{(i-1)!! j!! k!! (1+2p) (3 + 2p)} & \text{if \(i\) is odd, \(j\) is even, \(k\) is even, and \(i,j,k \in \ZZ\)} \\
            0 & \text{otherwise}
        \end{cases}
    \end{equation}
    \begin{equation}
        u_{i,jk}(a,b,c) = \begin{cases}
            a^i b^j c^k + a^i b^k c^j & \text{if } k \neq j \\
            a^i b^j c^j & \text{otherwise}
        \end{cases}\text{.}
    \end{equation}
\end{lemma}

\begin{lemma}\label{lemma_can_use_ext_points_in_min}
    Let \(\calP \subseteq \R^3\) be some convex polytope, \(v = [a,b,c]^\sfT \in \R^3\) a vector, \(V(a,b,c) \in \R^3\) another vector which may depend on \(v\), and \(k_1,k_2,\rho\) constants such that \(\rho \in [0,1]\), \(k_1,k_2 \geq 0\). Moreover, define \(t(x) \coloneq 1 + \braket{x,v}\) and \(s(x) \coloneq 1 + \braket{x, V(a,b,c)}\) such that \(t(x) \geq 0, s(x) \geq 0\) for all \(x \in \calP\). Then, for
    \begin{equation}
        \alpha(a,b,c) \coloneq \min_{x \in \calP} \frac{\max\left(s(x), \left(k_1 + k_2 \max\left(0,t(x) - \frac{1}{2}\right)\right)\right)}{t(x)}\text{,}
    \end{equation}
    we have that
    \begin{equation}
        \alpha(a,b,c) \geq \max \left\{\min_{x \in \calP^{\mathsc{ext}}} \frac{\rho\, s(x) + (1-\rho) \left(k_1 + k_2 \left(t(x) - \frac{1}{2}\right)\right)}{t(x)}, \min_{x \in \calP^{\mathsc{ext}}} \frac{\rho\, s(x) + (1-\rho) k_1}{t(x)} \right\}\text{,}
    \end{equation}
    where \(\calP^{\mathsc{ext}}\) are the extreme points of \(\calP\).
\end{lemma}

\begin{proof}
    This proof follows similar steps as were done in \cite[Lemma 28]{parekh_beating_2021}. We decompose \(x = \sum_i \lambda_i x_i\), where \(x_i \in \calP^{\mathsc{ext}}\), \(\lambda_i \geq 0\) for all \(i\) and \(\sum_i \lambda_i = 1\). We then use that fact that \(s\) and \(t\) are linear, that \(\max(x,y) \geq \rho\, x + (1-\rho)y\) for any $\rho \in [0,1]$, and \(\max\left(0,t(x) - \frac{1}{2}\right) \geq t(x) - \frac{1}{2}\).
    \begin{align}
        &\frac{\max\left(s(x), \left(k_1 + k_2 \max\left(0,t(x) - \frac{1}{2}\right)\right)\right)}{t(x)} \nonumber\\
        &\hspace{0.5in}\geq \frac{\rho\, s(x) + (1-\rho) \left(k_1 + k_2 \left(t(x) - \frac{1}{2}\right)\right)}{t(x)} \label{eq_sdjnuedmiosdca}\\
        &\hspace{0.5in}= \frac{\sum_i \lambda_i \left(\rho\, s(x_i) + (1-\rho) \left(k_1 + k_2 \left(t(x_i) - \frac{1}{2}\right)\right)\right)}{\sum_i \lambda_i t(x_i)} \\
        &\hspace{0.5in}\geq \min_{i : t(x_i) \neq 0} \frac{\rho\, s(x_i) + (1-\rho) \left(k_1 + k_2 \left(t(x_i) - \frac{1}{2}\right)\right)}{t(x_i)}
    \end{align}
    Similar steps can be used to show the following, this time we use \(\max(0,t(x)-\frac{1}{2}) \geq 0\) in place of \cref{eq_sdjnuedmiosdca}.
    \begin{equation}
        \frac{\max\left(s(x), \left(k_1 + k_2 \max\left(0,t(x) - \frac{1}{d}\right)\right)\right)}{t(x)} \geq \min_{i : t(x_i) \neq 0} \frac{\rho\, s(x_i) + (1-\rho) k_1}{t(x_i)}
    \end{equation}
    Combining these two inequality finishes the proof.
\end{proof}

\begin{lemma}\label{lemma_PT21b_lemma_29}
    Fix constants \(k_1,k_2,k_3 \geq 0\), let \(\calS\), \(\calS^{\mathsc{ext}}\), \(v=[a,b,c]^\sfT\), and \(V(a,b,c)\), be defined as in \cref{lemma_from_PT21b_1,lemma_can_use_ext_points_in_min} with
    \begin{align}
        V(a,b,c) = \E_{(z,z')}\left[\begin{bmatrix}
            z_1 z_1'\\
            z_2 z_2'\\
            z_3 z_3'
        \end{bmatrix}\middle/\sqrt{(z_1^2 + z_2^2 + z_3^2)((z_1')^2 + (z_2')^2 + (z_3')^2)} \right]\text{,}
    \end{align}
    where the expectation is taken over \((z,z') \sim \calN(0,\Sigma_{6}(\operatorname{diag}(a,b,c)))\). Then, for \(\vec{1} = [1,1,1]^\sfT\), we have that
    \begin{equation}\label{eq_scajkbehcasewdcsd}
        \begin{split}
        &\min_{x \in \calS^{\mathsc{ext}}} \min_{v = [a,b,c]^\sfT \in \calS} \frac{k_1 (1 + \braket{x, V(a,b,c)}) + k_2 \left(\frac{1}{2} + \braket{x, v}\right) + k_3}{1 + \braket{x, v}} \\
        &\hspace{2in}=\min_{v = [a,b,c]^\sfT \in \calS} \frac{k_1 (1 - \braket{\vec{1}, V(a,b,c)}) + k_2 \left(\frac{1}{2} - \braket{\vec{1}, v}\right) + k_3}{1 - \braket{\vec{1}, v}}\text{.}
    \end{split}
    \end{equation}
\end{lemma}

This proof follows in similar steps as were done in \cite[Lemma 29]{parekh_beating_2021} (indeed they proved the \(k_2=0\) case). Similar to their proof, we will prove something slightly stronger: that for any choice of \(x \in \calS^{\mathsc{ext}}\) the minimization over \((a,b,c) \in \calS\) gives the same value. The choice of \(x = -\vec{1}\) in \cref{eq_scajkbehcasewdcsd} is then somewhat arbitrary.

\begin{proof}

    
    It is evident that for a fixed \(x \in  \calS^{\mathsc{ext}}\), permutation of the entries doesn't change the value of the minimization over \((a,b,c)\in \calS\) as we can permute the values of \((a,b,c)\) to get back the original value. This is to say that without loss of generality, in \cref{eq_scajkbehcasewdcsd} it suffices to consider only the minimization over \(x \in \{[-1,-1,-1]^\sfT,[1,1,-1]^\sfT\}\).

    Consider the linear map \(\varphi : [x_1,x_2,x_3]^\sfT \mapsto [-x_1,-x_2,x_3]^\sfT\). We then use three main facts: First that \(\calS\) is invariant under the linear map \(\varphi : [x_1,x_2,x_3]^\sfT \mapsto [-x_1,-x_2,x_3]^\sfT\) (that is \(\varphi \calS = \calS\)), second that \(\varphi\) is an isometry under the dot product, and third that \(\E_{(x,x') \in \calN(0,\Sigma_2(a))}[x x'] = -\E_{(y,y') \in \calN(0,\Sigma_2(-a))}[y y']\), which means that \(\varphi V(a,b,c) = V(-a,-b,c)\). We can then do the following for any fixed \(x \in \calS^{\mathsc{ext}}\):
    \begin{align}
        &\min_{v = [a,b,c]^\sfT \in \calS} \frac{k_1 (1 + \braket{x, V(a,b,c)}) + k_2 \left(\frac{1}{2} + \braket{x, v}\right) + k_3}{1 + \braket{x, v}} \nonumber\\
        &\hspace{0.5in}= \min_{v = [a,b,c]^\sfT \in \calS} \frac{k_1 (1 + \braket{\varphi x, \varphi V(a,b,c)}) + k_2 \left(\frac{1}{2} + \braket{\varphi x, \varphi v}\right) + k_3}{1 + \braket{\varphi x, \varphi v}} \\
        &\hspace{0.5in}= \min_{v = [-a,-b,c]^\sfT \in \calS} \frac{k_1 (1 + \braket{\varphi x, V(-a,-b,c)}) + k_2 \left(\frac{1}{2} + \braket{\varphi x, v}\right) + k_3}{1 + \braket{\varphi x, v}} \\
        &\hspace{0.5in}= \min_{v = [a,b,c]^\sfT \in \calS} \frac{k_1 (1 + \braket{\varphi x, V(a,b,c)}) + k_2 \left(\frac{1}{2} + \braket{\varphi x, v}\right) + k_3}{1 + \braket{\varphi x, v}} \\
    \end{align}
    Because all points in \(\calS^{\mathsc{ext}}\) are related by permutation and \(\varphi\), we have shown that \cref{eq_scajkbehcasewdcsd} is the same for all choices of \(x \in \calS^{\mathsc{ext}}\).
\end{proof}

We can now bound the approximation ratio.

\begin{theorem}[Computational]\label{thm_bigger_number_for_proj_hams}
    For any instance of the \algprobm{Maximal Entanglement} problem for \(d=2\), the algorithm which runs \cref{algo_matching_algo_for_QUG} and \cite[Algorithm 9]{parekh_beating_2021}, returning the state with the larger energy, has an approximation guarantee of \(0.595\) on the energy of the Hamiltonian. 
\end{theorem}

\begin{proof}
    We follow the analysis of \cite[Theorem 11]{lee_improved_2024}. We consider the maximum energy from the matching \cref{algo_matching_algo_for_QUG} and the product state given by \cite[Algorithm 9]{parekh_beating_2021} to get the following for the bound on the approximation ratio.
    \begin{equation}
        \alpha_2 \geq \min_{\substack{(a,b,c) \in \calS \\ (p,q,r) \in \calS}} \max\left(\cref{eq_for_approx_ratio_numerator}, \frac{1}{4} + \frac{6}{5}\max\left(0,\cref{eq_for_approx_ratio_denominator} - \frac{1}{2}\right)\right)\Bigg/ \cref{eq_for_approx_ratio_denominator}
    \end{equation}
    Using \cref{lemma_from_PT21b_2}, define\footnote{We cut off \(i,j,k \leq 70\) so we can computationally implement \(\Tilde{s}\).} 
    \begin{equation}
        \Tilde{s}(a,b,c) = 1 - \sum_{\substack{i,j,k \leq 70\\j \leq k}} \hat{f}^2_{i,jk}\left(u_{i,jk}(a,b,c) + u_{i,jk}(b,a,c) + u_{i,jk}(c,a,b)\right)\text{,}
    \end{equation}
    which using the results of \cref{lemma_from_PT21b_2,lemma_can_use_ext_points_in_min,lemma_PT21b_lemma_29} and the fact that \(\text{\cref{eq_for_approx_ratio_numerator,eq_for_approx_ratio_denominator}} \geq 0\), gives us that
    \begin{align}
        \alpha_2 &\geq \max_{\rho \in [0,1]}\min_{(a,b,c) \in \calS} \max\left\{ \frac{\rho (\Tilde{s}(a,b,c) + \operatorname{rem}) + (1-\rho)\left(\frac{1}{4} + \frac{6}{5}(\frac{1}{2} - a - b - c)\right)}{1 - a - b - c}, \frac{\rho (\Tilde{s}(a,b,c) + \operatorname{rem}) + (1-\rho)\frac{1}{4}}{1 - a - b - c}\right\} \\
        \begin{split}
            &\geq \max_{\rho \in [0,1]} \min\Bigg\{\min_{\substack{-1 \leq a \leq b \leq c \leq 1\\a+b+c \leq \frac{1}{2}}} \frac{\rho (\Tilde{s}(a,b,c) + \operatorname{rem}) + (1-\rho)\left(\frac{1}{4} + \frac{6}{5}(\frac{1}{2} - a - b - c)\right)}{1 - a - b - c}, \\
            &\phantom{\geq \max_{\rho \in [0,1]} \min\Bigg\{} \min_{\substack{-1 \leq a \leq b \leq c \leq 1\\\frac{1}{2} \leq a+b+c \leq 1}} \frac{\rho (\Tilde{s}(a,b,c) + \operatorname{rem}) + (1-\rho)\frac{1}{4}}{1 - a - b - c} \Bigg\}\text{.}
            \end{split}
    \end{align}
    Here \(\operatorname{rem}\) denotes the higher order terms in the Hermite expansion, which are ignored computationally. The second inequality follows from the fact that \(S\) is invariant under permutation of the entries, so we can without loss of generality apply an ordering. The \(a+b+c \leq 1\) bound comes from the fact that \(\text{\cref{eq_for_approx_ratio_denominator}} \geq 0\). 

    Fixing \(\rho = 0.6724\), we computationally find the minimum over \(a,b,c\) to be \(\alpha_2 \geq 0.5957\), achieved at \(a=b=c\approx-0.4402\). Note, these match the optimal constants for \algprobm{Quantum Max-Cut} \cite{lee_improved_2024}. 
\end{proof}

\section{Open Problems}

We summarize the main problems left open by this work.
\begin{itemize}
    \item It is likely that one can improve on the \(4/5\) constant in \cref{lemma_lee_sdp_matching_lemma} with additional certificates as we did in \cref{lemma_Kn_matching_bound,lemma_K5_matching_bound_computational} (see also \cite{edmonds_maximum_1965,lee_improved_2024}).
    \item As noted in \cref{section_matching_alg}, the main gap in our analysis is the high degree setting. One could presumably generalize \cite{parekh_beating_2021} for qudits and combine the two algorithms in a similar way as was done in \cite{lee_improved_2024} and \cref{section_the_qubit_case}. This will likely require a more directed consideration towards degree-one terms than was previously done (see \cite[Appendix C]{carlson_approximation_2023}).
    \item Can anything be said about \algprobm{Maximal Entanglement} over expander graphs? Even for specific instances such as \algprobm{Quantum Max-Cut}, this is an open question. Our work gives some evidence that entanglement might not depend on the expansion of the interaction graph but instead merely on the degree. 
    \item To the best of our knowledge, the hardness of \algprobm{Maximal Entanglement} for fixed \(d\) is not known. In particular, \cite[Theorem 5]{piddock_universal_2021} allows for negative weights, unlike \algprobm{Maximal Entanglement}. However, if one relaxes the local Hamiltonian terms in \algprobm{Maximal Entanglement} from being maximally entangled pure states to maximally entangled mixed states (i.e., mixed state with partial traces giving the maximally mixed state), then \cite[Theorem 5]{piddock_universal_2021} can be used to argue that, for fixed \(d\), \algprobm{Mixed-State Maximal Entanglement} is \cc{QMA}-hard (as \(\frac{1}{\Tr(I-P)}(I - P)\) is a maximally entangled mixed state).
\end{itemize}

\section*{Acknowledgments}\label{sec_ack}
\addcontentsline{toc}{section}{\nameref*{sec_ack}}

We thank the anonymous reviewers for pointing out a typo in the statement of \cref{lemma_PT21b_lemma_29}, a typo in the proof of \cref{lemma_anticom_of_epr_give_sum_and_proj}, and for helping improve the presentation of much of the manuscript.

\printbibliography[heading=bibintoc]

\appendix
\part*{Appendix}
\addcontentsline{toc}{part}{Appendix}

\section{Proof of \texorpdfstring{\cref{prop_full_rank_states}}{Proposition~\ref{prop_full_rank_states}}}\label{appendix_proof_of_prop_full_rank_states}

\begin{manualproposition}{\ref{prop_full_rank_states}}[Restatement] Let \(\ket{\psi} \in (\C^d)^{\otimes 2}\) be a bipartite state with full Schmidt rank. 
    \begin{enumerate}
        \item There exists a unique \(A \in GL(d)\) such that \(\ket{\psi} = (I \otimes A) \ket{\EPR}\). In particular, if \(\ket{\psi}\) is maximally entangled then \(A \in \calU(d)\).
        
        \item For any \(A \in GL(d)\) we have that there exists a unique \(B \in GL(d)\) such that \((I \otimes A) \ket{\psi} = (B \otimes I) \ket{\psi}\) and vice versa. In particular, if \(\ket{\psi}\) is maximally entangled and \(A \in \calU(d)\) then \(B \in \calU(d)\). Additionally, if \(\ket{\psi} = \ket{\EPR}\) then \(B=A^{\sfT}\).
    \end{enumerate}
\end{manualproposition}

\begin{proof}
    Let \(\{\ket{e_1}, \dotsc, \ket{e_d}\}\) and \(\{\ket{f_1}, \dotsc, \ket{f_d}\}\) be the orthonormal bases for \(\C^d\) in accordance with the Schmidt decomposition, \(\ket{\psi} = \sum_{i = 1}^d \sqrt{\lambda_i} \ket{e_i} \otimes \ket{f_i}\).
    Let \(A : \ket{i} \mapsto \sqrt{d \lambda_i} \ket{e_i}\) and \(B : \ket{i} \mapsto \ket{f_i}\) be changes of basis, for which, we note that \(B \in \calU(d)\) and, in fact, by potentially absorbing phases into \(A\), we can make \(B \in \SU(d)\). It is easy to verified that \(\ket{\psi} = (A \otimes B) \ket{\EPR}\). When \(\ket{\psi}\) is maximally entangled we additionally have that \(\lambda_i = \frac{1}{d}\) and thus \(A \in \calU(d)\).

    Let \(V =\C^d\) denoted a \(d\)-dimensional vector space. We consider the standard bilinear form, \((\cdot,\cdot) \in V^* \otimes V^*\) defined by \((v_i, v_j) = \delta_{ij}\) for any orthonormal basis, \(\{v_i\}_i\)\footnote{We use this notation and subsequently \(v^*\) for dual vectors as well as \((\cdot, \cdot) : V \times V \ra \C\) to denote the bilinear dot product to avoid confusion with the conjugate linear map \(\ket{\psi} \mapsto \bra{\psi}\), in line with the bra-ket notation, and corresponding Hermitian inner-product \(\bra{\psi}\ket{\varphi}\).}. With this we can define the isomorphism from \(V\) to its dual space, \((\cdot)^* : V \ra V^*, v \mapsto v^*\), defined by \(v^*(w) = (v,w)\) for all \(v,w \in V\) and additionally, \((\cdot)^* : V^{\otimes 2} \ra (V^{\otimes 2})^*, v \otimes w \mapsto v^* \otimes w^*\) for all \(v,w \in V\), extended linearly. Viewing \(\ket{\psi} \in V \otimes V\) as a 2-tensor, we can then identify \(\ket{\psi}\) with a bilinear form, defined by \(\ket{\psi}^*(v,w) = (\ket{\psi}, v \otimes w)\) for any \(v,w \in V\), where \((\cdot, \cdot) : V^{\otimes 2} \times V^{\otimes 2} \ra \C\) is also used to denote the bilinear dot product over \(V^{\otimes 2}\). For a linear operator \(A \in \calL(V)\), we use \(A^* \in \calL(V^*)\) to denote its dual map defined by \(A^* f = f \circ A\) for \(f \in V^{*}\).
    
    Because the map \((\cdot)^*\) is an isomorphism, to show equality between two states, \(\ket{\psi}\) and \(\ket{\varphi}\), it suffices to show equality between their corresponding dual vectors. Moreover, the bilinear form, \(\ket{\psi}^*\), is non-degenerate because the state, \(\ket{\psi} = (A \otimes B) \ket{\EPR}\), is full rank by assumption. This will allow us to define an adjoint in \(\ket{\psi}^*\). First, for \(\ket{\EPR}\), we can observe that the associated bilinear form, \(\ket{\EPR}^*\), is nothing but the scaled bilinear dot, \(\ket{\EPR}^*(v,w) = \frac{1}{\sqrt{d}}(v,w)\). This can be verified through direct calculation. It is then evident that \(\ket{\psi}^* = ((A \otimes B)\ket{\EPR})^*\) is the bilinear form defined by the following, where \(A^{\sfT}\) is taken to be the adjoint over the bilinear dot product, i.e., the transpose. 
    \begin{align}
        \ket{\psi}^*(v,w) &= ((A \otimes B)\ket{\EPR})^*(v \otimes w) \\
        &= ((A \otimes B) \ket{\EPR}, v \otimes w) \\
        &= (\ket{\EPR}, A^{\sfT} v \otimes B^{\sfT} w) \\
        &= \ket{\EPR}^* (A^{\sfT} v, B^{\sfT} w) \\ 
        &= \frac{1}{\sqrt{d}}(A^{\sfT} v, B^{\sfT} w) \\ 
        &= \frac{1}{\sqrt{d}}(v, A B^{\sfT} w)\\ 
        &= ((I \otimes BA^{\sfT})\ket{\EPR})^*(v, w)
    \end{align}
    Thus, \(\ket{\psi} = (I \otimes BA^{\sfT}) \ket{\EPR}\), which proves \cref{prop_item_full_rank_states_2}. Uniqueness follows from the fact that \((I \otimes A) \ket{\EPR} = (I \otimes B) \ket{\EPR} \Leftrightarrow \forall v,w \in V : (v, Aw) = (v, Bw) \Leftrightarrow A = B\).

    Next, using the fact that \(\ket{\psi}^*(v,w) = \frac{1}{\sqrt{d}}(v, A B^{\sfT} w)\) we have the following.
    \begin{align}
        \left((C \otimes I) \ket{\psi}\right)^*(v,w) &= \ket{\psi}^*(C^{\sfT} v,w) \\
        &= \frac{1}{\sqrt{d}}(v, C A B^{\sfT} w) \\
        &= \ket{\psi}^*(v, (AB^{\sfT})^{-1} C AB^{\sfT} w) \\
        &= \left((I \otimes ((AB^{\sfT})^{-1} C AB^{\sfT})^{\sfT}) \ket{\psi}\right)^*(v,w)
    \end{align}
    Here, \((AB^{\sfT})^{-1} C AB^{\sfT}\) is the adjoint of \(C^\sfT\) under the bilinear form \(\ket{\psi}^*\).
    All together, we have that \((C \otimes I)\ket{\psi} = \left(I \otimes \left(BA^{\sfT} C^{\sfT} (BA^{\sfT})^{-1}\right)\right) \ket{\psi}\).
    Similarly, we also have that \((I \otimes C) \ket{\psi} = \left(\left(AB^{\sfT} C^{\sfT} (A B^{\sfT})^{-1}\right) \otimes I\right) \ket{\psi}\).
 
    The uniqueness follows from the uniqueness of the adjoint over non-degenerate bilinear forms.

    When \(\ket{\psi}\) is maximally entangled and \(C \in \calU(d)\), then \(A,B \in \calU(d)\), as established earlier. Finally, it follows that \(BA^{\sfT} C^{\sfT} (BA^{\sfT})^{-1} = BA^{\sfT} C^{\sfT} AB^{\sfT} \in \calU(d)\) as \(C^{\sfT} \in \calU(d)\).
\end{proof}

\section{Proof of \texorpdfstring{\cref{lemma_anticom_of_epr_give_sum_and_proj}}{Lemma~\ref{lemma_anticom_of_epr_give_sum_and_proj}}}\label{appendix_proof_of_lemma_anticom_of_epr_give_sum_and_proj}

\begin{manuallemma}{\ref{lemma_anticom_of_epr_give_sum_and_proj}}[Restatement]
    There exists a projector \(P \in \calL((\C^d)^{\otimes 3})\) (i.e., \(P^2 = P\)) such that
    \begin{equation}
        \{\EPR \otimes I, I \otimes \EPR\} = \frac{1}{d}(\EPR \otimes I + I \otimes \EPR) - \frac{2(d-1)}{d^2} P\text{.}
    \end{equation}
\end{manuallemma}

\begin{proof}
    We can write the EPR projector in the standard basis as follows:
    \begin{equation}
        \EPR \otimes I = \frac{1}{d}\sum_{a,b = 1}^d \ket{aa}\bra{bb} \otimes \sum_{c = 1}^d \ket{c}\bra{c}
    \end{equation}
    And similarly:
    \begin{equation}
        I \otimes \EPR = \sum_{a = 1}^d \ket{a}\bra{a} \otimes \frac{1}{d}\sum_{b,c = 1}^d \ket{bb}\bra{cc}
    \end{equation}
    Then, we consider \(P\):
    \begin{align}
        P &= \frac{d^2}{2(d-1)}\left(\frac{1}{d}(\EPR \otimes I + I \otimes \EPR) - \{\EPR \otimes I, I \otimes \EPR\}\right) \\
        &= \frac{d^2}{2(d-1)}\left(\frac{1}{d^2}\left(\sum_{a,b,c = 1}^d \ket{aac}\bra{bbc} + \ket{abb}\bra{acc}\right) - \left((\EPR \otimes I)(I \otimes \EPR) + (I \otimes \EPR)(\EPR \otimes I)\right)\right) \\
        &= \frac{d^2}{2(d-1)}\left(\frac{1}{d^2}\left(\sum_{a,b,c = 1}^d \ket{aac}\bra{bbc} + \ket{abb}\bra{acc}\right) - \frac{1}{d^2}\left(\sum_{a,b,c = 1}^d\sum_{x,y,z = 1}^d  \ket{aac}\bra{bbc}\ket{xyy}\bra{xzz} + d^2(I \otimes \EPR)(\EPR \otimes I)\right)\right) \\
        &= \frac{1}{2(d-1)}\left(\left(\sum_{a,b,c = 1}^d \ket{aac}\bra{bbc} + \ket{abb}\bra{acc}\right) - \left(\sum_{a,b,z = 1}^d  \ket{aab}\bra{bzz} + \sum_{a,b,y = 1}^d  \ket{abb}\bra{yya}\right)\right) \\
        &= \frac{1}{2(d-1)}\left(\sum_{a,b,c = 1}^d \underbrace{\ket{aac}\bra{bbc}}_{S_1} + \underbrace{\ket{abb}\bra{acc}}_{S_2} -  \underbrace{\ket{aab}\bra{bcc}}_{S_3} - \underbrace{\ket{abb}\bra{cca}}_{S_4}\right) \label{eq_P_broke_up_into_4_things}
    \end{align}
    Next, we calculate \(P^2\). We do this in multiple parts by considering the products of the \(4\) parts of \cref{eq_P_broke_up_into_4_things}. Note, we will always refer the parts as positive sums, that is, for example, we let \(S_3 = \sum_{a,b,c = 1}^d \ket{aab}\bra{bcc}\), and in particular, \(P = \frac{1}{2(d-1)}\left(S_1 + S_2 - S_3 - S_4\right)\).
    
    First, we consider the square of each part. The first two, being nothing but scaled projectors onto the EPR state, give \(S_1^2 = dS_1\) and \(S_2^2 = dS_2\). We then consider:
    \begin{align}
        S_3^2 &= \left(\sum_{a,b,c = 1}^d \ket{aab}\bra{bcc}\right) \left(\sum_{x,y,z = 1}^d \ket{xxy}\bra{yzz}\right) \\
        &= \sum_{a,b,c = 1}^d \sum_{x,y,z = 1}^d \ket{aab}\bra{bcc}\ket{xxy}\bra{yzz} \\
        &= \sum_{a,b = 1}^d \sum_{z = 1}^d \ket{aab}\bra{bzz} \\
        &= S_3
    \end{align}
    
    By symmetry, we also have that \(S_4^2 = S_4\). For the cross terms, first note that \(\{S_1, S_2\} = d^2\{\EPR \otimes I, I \otimes \EPR\} = S_3 + S_4\), as shown above in \cref{eq_P_broke_up_into_4_things}. Next, we consider:
    \begin{align}
        S_1 \cdot S_3 &= \left(\sum_{a,b,c = 1}^d \ket{aac}\bra{bbc}\right) \left(\sum_{x,y,z = 1}^d \ket{xxy}\bra{yzz}\right) \\
        &= \sum_{a,b,c = 1}^d \sum_{x,y,z = 1}^d \ket{aac}\bra{bbc}\ket{xxy}\bra{yzz}\\
        &= \sum_{a,b,c = 1}^d \sum_{z = 1}^d \ket{aac}\bra{czz} \\
        &= d \sum_{a,c,z = 1}^d \ket{aac}\bra{czz} \\
        &= d S_3 
    \end{align}
    
    By symmetry, we also have that \(S_4 \cdot S_1 = dS_4\).
    Also:
    \begin{align}
        S_3 \cdot S_1 &= \left(\sum_{a,b,c = 1}^d \ket{aab}\bra{bcc}\right) \left(\sum_{x,y,z = 1}^d \ket{xxz}\bra{yyz}\right) \\
        &= \sum_{a,b,c = 1}^d \sum_{x,y,z = 1}^d \ket{aab}\bra{bcc}\ket{xxz}\bra{yyz}\\
        &= \sum_{a,b,y = 1}^d \ket{aab}\bra{yyb} \\
        &= S_1 
    \end{align}
    
    By symmetry, we also have that \(S_1 \cdot S_4 = S_1\). Additionally, for similar reasons, \(\{S_2, S_3\} = S_2 + dS_3\) and \(\{S_2, S_4\} = S_2 + dS_4\). Finally: 
    \begin{align}
        S_3 \cdot S_4 &= \sum_{a,b,c = 1}^d \sum_{x,y,z = 1}^d \ket{aab}\bra{bcc}\ket{xyy}\bra{zzx}\\
        &= \sum_{a,b,c = 1}^d \sum_{z = 1}^d \ket{aab}\bra{zzb} \\
        &= d S_1
    \end{align}
    
    And similarly, we also have that \(S_4 \cdot S_3 = d S_2\).
    Finally, putting it all together, we get the following for \(P^2\):
    \begin{align}
        P^2 &= \frac{1}{4(d-1)^2}\left(S_1^2 + S_2^2 + S_3^2 + S_4^2 + \{S_1,S_2\} - \{S_1,S_3\} - \{S_1,S_4\} - \{S_2,S_3\} - \{S_2,S_4\} + \{S_3,S_4\}\right) \\
        &= \frac{1}{4(d-1)^2}\left(dS_1 + dS_2 + S_3 + S_4 + S_3 + S_4 - S_1 - dS_3 - S_1 - dS_4 - S_2 - dS_3 - S_2 - dS_4 + dS_1 + dS_2\right) \\
        &= \frac{(2d-2)}{4(d-1)^2}\left(S_1 + S_2 - S_3 - S_4\right) \\
        &= \frac{1}{2(d-1)}\left(S_1 + S_2 - S_3 - S_4\right) \\
        &= P
    \end{align}
    Thus, \(P\) is a projector.
\end{proof}

\end{document}